\documentclass[12pt]{amsart}
\usepackage{amssymb}
\usepackage{amsmath}
\allowdisplaybreaks[1]
\usepackage{graphicx,color}
\usepackage{wrapfig,framed}

\usepackage[height=23cm, width=16.0cm, hmarginratio={1:1}]{geometry}
\usepackage{hyperref}

\usepackage{appendix}

\newtheorem{theorem}{Theorem}
\newtheorem{prop}{Proposition}
\newtheorem{lemma}{Lemma}
\newtheorem{cor}{Corollary}
\newtheorem*{remark}{Remark}

\newcommand{\beq}{\begin{equation}}
\newcommand{\eeq}{\end{equation}}
\newcommand{\nn}{\nonumber}

\newcommand{\F}{\mathcal{F}}

 \newcommand{\bP}{{\mathbb P}}
\newcommand{\RR}{{\mathbb R}}
\newcommand{\CC}{{\mathbb C}} \newcommand{\cH}{{\mathcal H}} \newcommand{\cM}{{\mathcal M}}
\newcommand{\cP}{{\mathcal P}}
\newcommand{\ZZ}{{\mathbb Z}} 

\newcommand{\tr}{{\rm tr}}

\newcommand{\e}{\epsilon}
\newcommand{\p}{\partial}

\newcommand{\half}{\frac{1}{2}}
\newcommand{\vac}{|0\rangle}

\newcommand{\corr}[1]{\langle {#1} \rangle}

\DeclareMathOperator{\Img}{Im} \DeclareMathOperator{\vol}{vol}

\newcommand{\ben}{\begin{eqnarray*}}
\newcommand{\een}{\end{eqnarray*}}
\newcommand{\bea}{\begin{eqnarray}}
\newcommand{\eea}{\end{eqnarray}}

\begin{document}

\title[Grothendieck's Dessins D'Enfants in a Web of Dualities]
{Grothendieck's Dessins d'Enfants in a Web of Dualities. III.}
\author{Di Yang, Jian Zhou}
\date{}
\maketitle

\begin{abstract}
We identify the dessin partition function with the partition function of the Laguerre unitary 
ensemble (LUE). Combined with the result due to Cunden {\em et al} on the relationship between 
the LUE correlators and strictly monotone Hurwitz numbers introduced by Goulden {\em et al}, 
we then establish connection of dessin counting to strictly monotone Hurwitz numbers. 
We also introduce a correction factor for the dessin/LUE partition function, which plays an important role in
showing that the corrected dessin/LUE partition function is a tau-function
of the Toda lattice hierarchy. 
As an application, we use the
approach of Dubrovin and Zhang for the computation of the dessin correlators. 
In physicists' terminology,
we establish dualities among dessin counting, generalized Penner model, and $\bP^1$-topological sigma model.
\end{abstract}

\setcounter{tocdepth}{1}
\tableofcontents

\section{Introduction}\label{section1}

In this sequel to \cite{Zhou2,Zhou3},
we will continue to study the duality of enumerations of Grothendieck's dessins
d'enfants with some other theories.
The proposal in \cite{Zhou2} is to view the partition functions of
different theories as tau-functions of the KP hierarchy,
and use the fact that the Sato Grassmannian, the space of tau-functions,
is an infinite-dimensional homogeneous space under the action of
$\widehat{GL}(\infty)$.
In this paper,
we will study the duality between dessin partition function with partition
functions of other theories from the viewpoint of Toda lattice hierarchy.

In this new approach to duality,
we can apply the theory of normal forms of integrable hierarchies and the extended Toda hierarchy~\cite{CDZ}
as developed by Dubrovin and Zhang~\cite{DZ-norm, DZ} to establish connections
with other theories.
In this theory,
one starts with a semisimple Frobenius manifold and constructs 
a hierarchy in genus zero associated with it, 
called the {\em principal hierarchy}. 
This will determine the free energy in genus zero from some initial values.
The higher genera parts of the free energy are determined recursively by Virasoro constraints
in a form called {\em loop equation} in \cite{DZ-norm}.
Note that in this theory the free energy in each genus is expressed in terms of 
the jet variables on the loop space of the Frobenius manifold; 
these expressions are completely determined by the Frobenius manifold,
and {\it do not} depend on the initial values in genus zero.
We interpret this as giving a universality class of criticality in the renormalization 
theory of quantum field theories. 

The integrable hierarchy associated with the Frobenius
manifold underlying the GW theory of~$\bP^1$
(also known as the $\bP^1$-topological sigma model) is
the extended Toda hierarchy \cite{CDZ,DZ-norm,DZ,Zhang}.
It is just the Toda lattice hierarchy with an extra series of commuting flows.
So our results establish a connection to all theories whose partition functions are 
tau-functions of the (extended) Toda lattice hierarchy.

Before we give more detailed statements of our results, 
let us recall some definitions and  previous results from \cite{Zhou2}. 
Let $(C,f)$ be a Belyi pair of genus~$g$ and degree~$d$, and $\Gamma$ the corresponding dessin. 
(For applications of Belyi pairs in physics see e.g.~\cite{BFHHRXY}.)
Put $k=|f^{-1}(0)|$, $l=|f^{-1}(1)|$, and $m=|f^{-1}(\infty)|$. By the Riemann-Hurwitz formula we know that
\beq\label{riemannhurwitz}
2g-2 = d-(k+l+m).
\eeq
We assume that the poles of~$f$ are labeled and denote the set of their orders by $\mu=(\mu_1,\dots,\mu_m)$ so
that $d=\sum_{i\ge1} \mu_i$. The tripe $(k,l,\mu)$ will be called the {\it type} of the dessin~$\Gamma$, and the
set of all dessins of type $(k,l,\mu)$ will be denoted by $\mathcal{D}_{k,l;\mu}$.
Let $N_{k,l}(\mu)=N_{k,l}(\mu_1,\dots,\mu_m)$ denote the weighted count of the labeled dessins d'enfants, i.e.,
\beq\label{deinklmum}
N_{k,l}(\mu_1,\dots,\mu_m) := \sum_{\Gamma\in \mathcal{D}_{k,l,\mu}}  \frac1{|{\rm Aut}_b \Gamma|} ,
\eeq
where ${\rm Aut}_b \Gamma$ denotes the group of automorphisms of~$\Gamma$ that preserve the boundary component wise.
For $m\geq1$, and $\mu_1,\dots,\mu_m\geq1$, define the
{\it connected %Grothendieck's
$m$-point dessin correlator} as follows:
\begin{align}
\langle \tau_{\mu_1}\cdots\tau_{\mu_m} \rangle (u,v,\e)
%& = \sum_{g\geq0}  \e^{2g-2}  \sum_{k,l\geq1}  u^k  v^l  \sum_{|\mu|-(k+l+m)=2g-2}  N_{k,l}(\mu)  \nn\\
& = \sum_{g\geq0, \, k,l\ge1 \atop 2g+k+l=|\mu|-m+2}  N_{k,l}(\mu)  \, u^k  v^l  \e^{2g-2}, \label{defdessinscorr}
\end{align}
where $u,v,\e$ are indeterminates, $|\mu|:=\mu_1+\dots+\mu_m$.

Let
$\F_{\rm dessins}=\F_{\rm dessins}(u,v,{\bf p};\e)$
be following the generating function of connected dessin correlators:
\begin{align}
\F_{\rm dessins}(u,v,{\bf p};\e)
& :=
%\sum_{k,l,m\geq1}  \frac{1}{m!}  \sum_{\mu_1,\dots,\mu_m\geq1}  N_{k,l}(\mu)  \e^{|\mu|-k-l-m}  u^k  v^l  p_{\mu_1} \cdots p_{\mu_m} \\& =
\sum_{m\geq1}
\sum_{\mu_1,\dots,\mu_m\geq1} \langle \tau_{\mu_1} \cdots \tau_{\mu_m} \rangle(u,v,\e) \, \frac{p_{\mu_1} \cdots p_{\mu_m}}{m!}, \label{deffdessins}
\end{align}
called the {\it dessin free energy}.
Here ${\bf p}=(p_1,p_2,p_3,\cdots)$ is an infinite vector of indeterminates.
The exponential
\beq \label{deffzduvpe}
e^{\F_{\rm dessins}(u,v,{\bf p};\e)} =: Z_{\rm dessins}(u,v,{\bf p};\e) = Z_{\rm dessins}
\eeq
is called the {\it dessin partition function}.

Based on the fact that the dessin partition function $Z_{\rm dessins}$ is a particular tau-function of the KP hierarchy
satisfying some Virasoro constraints~\cite{Zograf} (cf.~\cite{Zhou2, Zhou3}),
the second-named author of the present paper obtains~\cite{Zhou2}
the explicit affine coordinates~\cite{EH,Zhou} for the point in Sato Grassmannian
that corresponds to~$Z_{\rm dessins}$.
It then follows from the Theorem~5.3 of~\cite{Zhou}
an explicit formula for the generating series of $m$-point connected dessin correlators, namely, we have

\noindent {\bf Theorem A} (\cite{Zhou,Zhou2}). For each $m\geq2$,
the generating series for $m$-point connected dessin correlators
has the following expression:
\beq\label{connectedAformula}
\sum_{\mu_1,\dots,\mu_m\ge1} \prod_{j=1}^{m} \mu_j 
\frac{\langle\tau_{\mu_1}\cdots\tau_{\mu_m}\rangle(n,w,1)}
{\lambda_1^{\mu_1+1}\cdots \lambda_m^{\mu_m+1}} = (-1)^{m-1} \sum_{\sigma\in S_m/C_m}
\prod_{j=1}^m \widehat{A}(\lambda_{\sigma(j)},\lambda_{\sigma(j+1)}) - \frac{\delta_{m,2}}{(\lambda-\mu)^2},
\eeq
where $\widehat{A}(\lambda,\mu)$ is defined by
\beq\label{defhata}
\widehat{A}(\lambda,\mu):=\frac1{\lambda-\mu} + \sum_{i,j\geq0} \frac{A_{i,j}}{\lambda^{i+1}\mu^{j+1}}
\eeq
with
\beq\label{defaij}
A_{i,j}  :=  \frac{(-1)^j nw}{(i+j+1)i!j!} \prod_{r=1}^i (n+r)(w+r)  \prod_{r=1}^j (n-r)(w-r).
\eeq
We know from formula~\eqref{riemannhurwitz} that the $\e$-dependence in the connected dessin correlators can be reconstructed from
$\langle\tau_{\mu_1}\cdots\tau_{\mu_m}\rangle(n,w,1)$; in other words, all the
dessin counting numbers $N_{k,l}(\mu_1,\dots,\mu_m)$ with $m\geq2$ can be computed by~\eqref{connectedAformula}.
Moreover, the algorithm developed from~\eqref{connectedAformula} is efficient for large genus.

To relate the dessin partition function $Z_{\rm dessins}$ to the theory of Toda lattices (cf.~\cite{AvM, CDZ, Deift, DY1, DZ, UT}), 
we introduce, inspired also from the Gaussian unitary ensemble (GUE) partition function (cf.~\cite{DY1}) and the 
Laguerre unitary ensemble (LUE) model~\cite{AvM, CDOC, GGR}, 
 a certain correction factor to the dessin partition function. 
Explicitly, define a power series $Z(x,{\bf p};a;\e)$ by 
\beq\label{correctionfactordesssins}
Z(x,{\bf p};a;\e) := 
%(-1)^{\frac1{24}}
\e^{-\frac1{12} +\frac{v^2}{\e^2}} (2 \pi)^{-\frac{v}{\e}} \frac{G(1+\frac{u}\e)G(1+\frac{v}\e)}{G(1+\frac{v-u}\e)} Z_{\rm dessins}(u,v,{\bf p};\e).
\eeq
Here, 
$G(z)$ denotes the Barnes $G$-function, and the factors $G(1+\frac{u}\e)$, $G(1+\frac{v}\e)$, $G(1+\frac{v-u}\e)$ appearing
in~\eqref{correctionfactordesssins} are understood as the exponential of the $\e\to 0$
formal asymptotics of their logarithms, where we recall that as $z\to\infty$,
\beq \label{eqn:G-asymp}
\log G(1+z) \sim \frac{z^2}{2} \biggl(\log z-\frac32\biggr) + \frac{z}2 \log (2\pi) - \frac1{12} \log z
+ \zeta'(-1) + \sum_{\ell\geq1} \frac{B_{2\ell+2}}{4\ell(\ell+1) z^{2\ell}}.
\eeq
With respect to the KP times $p_1,p_2,\dots$, the correction factor is regarded as a constant. 
It is also convenient to work directly with formal power series~\cite{AIKZ}.
Indeed, we define the {\it corrected dessin free energy}~$\F=\F(x, {\bf p}; a; \e)$ by
\begin{align}
&\F(x, {\bf p}; a; \e) := \F_{\rm dessins}\bigl(x,x+a, {\bf p};\e\bigr) \label{2022cf} \\
& + \frac1{\e^2}  \biggl( \frac{x^2}2 \log x + \frac{(x+a)^2}{2}  \log (x+a)
- \frac{a^2}2 \log a-\frac{3}{2} x(x+a) \biggr) \nn\\
& - \frac1{12} \Bigl( \log x + \log (x+a\bigr) - \log a\Bigr) + \zeta'(-1)\nn\\
& + \sum_{g\ge2}  \frac{B_{2g}}{4g(g-1)}  \frac{\e^{2g-2}}{x^{2g-2}}
+ \sum_{g\ge2}  \frac{B_{2g}}{4g(g-1)}  \frac{\e^{2g-2}}{(x+a)^{2g-2}} - \sum_{g\ge2}  \frac{B_{2g}}{4g(g-1)}  \frac{\e^{2g-2}}{a^{2g-2}}, \nn
\end{align}
and we have $\exp(\F(x,{\bf p};a;\e))=Z(x,{\bf p}; a; \e)$. We call $Z=Z(x,{\bf p}; a; \e)$ the {\it corrected dessin partition function}.

It is known~\cite{KZ} (cf.~\cite{Zhou1}) that the dessin partition function 
$Z_{\rm dessins}(u,v,{\bf p};\e)$ satisfies Virasoro constraints 
(see~\eqref{virazdessins} of \S~\ref{newsection21121}). Then by 
using~\eqref{correctionfactordesssins} we immediately obtain that  
 the corrected dessin partition function~$Z$ satisfies the following Virasoro constraints:
\beq\label{virazfull}
L_k (Z) = 0 , \quad k\geq 0,
\eeq
where $L_k$ are linear operators given by 
\beq\label{virasorolueZ}
L_k = 2k \Bigl(x+\frac{a}2\Bigr) \frac{\p }{\p p_k} + \sum_{j\geq1} (k+j) \tilde{p}_j  \frac{\p }{\p p_{k+j}}
+ \frac1{\e^2}\sum_{i,j\ge1\atop i+j=k}  i  j  \frac{\p^2}{\p p_i \p p_j} + \frac{x(x+a)}{\e^2} \delta_{k,0}
\eeq
with $\tilde p_j:=p_j-\delta_{j,1}$, 
satisfying $[L_{k_1},L_{k_2}]=(k_1-k_2)L_{k_1+k_2}$ ($\forall\, k_1,k_2\geq0$). 
We will show in \S~\ref{newsection21121} that 
$Z$ also satisfies the following {\it dilaton equation}:
\beq\label{dilatonforZ121}
 \sum_{j\geq1}  \tilde p_j  \frac{\p Z}{\p p_{j}} + \e  \frac{\p Z}{\p \e} + x  \frac{\p Z}{\p x} +
a \frac{\p Z}{\p a}  + \frac1{12}Z = 0.
\eeq

Based on Theorem~A and on the matrix-resolvent (MR) method to tau-functions for the Toda lattice hierarchy~\cite{DY1}, 
we prove in \S~\ref{newsection3} the following theorem. 
\begin{theorem}\label{thm2}
The corrected dessin partition function $Z(x,{\bf p};a;\e)$ is a  
 tau-function for the Toda lattice hierarchy. In particular, the functions $V(x,{\bf p};\e)$ and $W(x,{\bf p};\e)$, 
 defined by
\begin{align}
 & V(x,{\bf p};\e) = \e (\Lambda-1) \frac{\p \log Z(x,{\bf p};a;\e)}{\p p_1}, \\
 & W(x,{\bf p};\e)= \frac{Z(x+\e,{\bf p};a;\e)Z(x-\e,{\bf p};a;\e)}{Z(x,{\bf p};a;\e)^2},
\end{align}
satisfy the Toda lattice hierarchy (see~\eqref{TLH}) with $t_i = p_{i+1}/(i+1)$, $i\geq0$.
Moreover, the solution $(V(x,{\bf s};\e), W(x,{\bf s};\e))$ is uniquely specified by the following initial data:
\beq\label{initialvaluetodadessins}
V(x,{\bf 0};\e) = 2 x+a+\e, \quad W(x, {\bf 0};\e) = x(x+a).
\eeq
\end{theorem}

Once a connection with Toda lattice hierarchy is established for the dessin counting,
one can use the techniques developed in the theory of the former to study the latter.
In this paper we will focus on the application,
as alluded to in the beginning of this Introduction,
of the theory of integrable hierarchies associated with semisimple Frobenius manifolds
as developed by Dubrovin and Zhang \cite{DZ-norm}.
The Lax operator for the Toda lattice hierarchy is a linear difference operator of the form:
\beq
L= \Lambda + V(x,{\bf p};\e) + W(x,{\bf p};\e) \Lambda^{-1},
\eeq
where $\Lambda:f(x)\mapsto f(x+\e)$ is the shift operator.
According to Theorem~\ref{thm2}, 
we can read from~\eqref{initialvaluetodadessins}
 the initial Lax operator~$L_{\rm ini}= \Lambda+ 2 x+a+\e + x(x+a) \Lambda^{-1}$ for the dessin solution.
The Toda lattice hierarchy can be thought of as a system of evolution equations on the $(V,W)$-plane.
It turns out that there is a structure of Frobenius manifold
on the $(v=V|_{\e=0}, u = \log(W)|_{\e=0})$-plane.
In fact the dispersionless extended Toda lattice hierarchy coincides with the principal hierarchy of
the semisimple Frobenius manifold with the Frobenius potential
\beq\label{todafrob}
F({\bf v})=\frac12 v^2 u + e^{u}
\eeq
and the Euler vector field
\beq\label{todaeuler}
E=v\p_v + 2 \p_u.
\eeq
This is the Frobenius manifold associated with the GW theory of $\mathbb{P}^1$.
As shown in \cite{DZ},
the integrable hierarchy associated with this Frobenius manifold is the extended Toda hierarchy.
We now have connection between
dessin counting and Gromow-Witten theory of~$\mathbb{P}^1$,
and in particular we can use the Dubrovin--Zhang method \cite{DY1,DZ-norm,DZ} for computing
 the dessin free energy; 
this gives another algorithm for computing the dessin correlators which is particularly efficient when $m$ is large.
The steps of this method are as follows.
First, the initial values of Toda lattice hierarchy
give a point in the Frobenius manifold.
By computing the Riemann invariants,
one sees whether one gets a {\em monotone solution} of the principal hierarchy associated to the above Frobenius manifold.
If the solution has monotonicity then in genus zero it can be obtained by the hodograph method.
In~\cite{DZ-norm,DZ}
the {\it quasi-triviality}~\cite{DZ} for the extended Toda hierarchy gives a construction which
transforms every flow of the principal hierarchy to the corresponding flow of the extended Toda hierarchy.
In particular,
they obtain some universal formulas in genus one \eqref{f1gw} and the loop equation \eqref{loop}
in genus bigger than one that compute the free energy in genus $\geq 1$ from the results in genus zero.

Note that formulas \eqref{f1gw} and \eqref{loop} are {\em universal} in the sense that they
do not depend on the initial values of $u$ and $v$ on the Frobenius manifold.
So we would like to interpret it as giving a universality class in some space of quantum field theories
as the critical point under some renormalization flow.
In the Dubrovin--Zhang method,
the jet variables on the Frobenius manifold are used in \eqref{f1gw} and \eqref{loop}.
From the point of view of~\cite{ZZ},
it should be possible to transform the jet variables to some renormalization variables.

So now we can combine the results and perspectives of this paper with our earlier work \cite{DLYZ, DY1,Zhou2}.
First of all,
by a result in~\cite{Zhou2},
the normalized modified GUE partition function with even couplings, up to a constant factor (can depend on $x$), 
can be identified with the dessin partition function,
and by the Hodge-GUE correspondence \cite{DLYZ, DY2},
it can also be identified with some series obtained from triple Hodge integrals.
Therefore,
by multiplying the correction factor \eqref{correctionfactordesssins} one can see that
these two kinds of partition functions also give tau-functions
of the  Toda  hierarchy and hence lie in the above universality class.
Secondly, it is known in \cite{DY1} that the GUE partition function with a suitable correction factor given by the Barnes function
is a tau-function of the Toda lattice hierarchy and so it lies in the same universality class.
Last, but not the least, the GW theory of~$\bP^1$~\cite{DZ} lies in this universality class.
The difference among these theories is only their initial values which give different
parametrized curves on a two-dimensional Frobenius manifold.

One recognizes from the initial data~\eqref{initialvaluetodadessins} 
the LUE solution to the Toda lattice hierarchy (cf.~\cite{AvM, GGR}). Following~\cite{AvM, CDOC, GGR}
define the {\it normalized LUE partition function} by
\beq\label{defzlue1intro}
Z_{\rm LUE1}(x,{\bf p};a;\e) = 
\frac{G(\alpha+1)  \e^{-n^2-\alpha n}}{\pi^{\frac{n(n-1)}{2}}  G(n+\alpha+1)}
\int_{\mathcal{H}^+_n}  (\det{M})^\alpha  e^{- \frac1\e {\rm tr} \, V(M; {\bf p}) }  dM,
\eeq
where $\mathcal{H}_n^+$ denote the space of positive hermitian matrices of size~$n$, 
$\alpha=a \epsilon$, $x=n \epsilon$,
\beq\label{vmdef}
V(M;{\bf p})= M-\sum_{j\ge1} \frac{p_j}{j} M^j,
\eeq
and
\beq\label{invmeasure}
dM = \prod_{1\leq i\leq n}  d M_{ii}  \prod_{1\leq i<j\leq n}  d {\rm Re} M_{ij} d{\rm Im} M_{ij}.
\eeq
%The following corollary follows immediately from~\eqref{zdessinsew1} and Proposition~\ref{Zoneweone1}.
By using Theorem~\ref{thm2} and a result in~\cite{GGR} we will prove in \S~\ref{newsection3} the following corollary.
\begin{cor}\label{corzdzo} We have
\beq\label{identityzdzo}
Z_{\rm dessins}(x,x+a,{\bf p};\e) = Z_{\rm LUE1}(x,{\bf p};a;\e).
\eeq
\end{cor}
We will also give another proof of this corollary by 
using the Virasoro constraints \cite{AvM, CDOC, GGR, HH, Zhou2, Zhou3, Zograf}. 
We note that matrix models for Grothendieck's dessin counting were suggested by 
Ambj{\o}rn and Chekhov~\cite{AC}.

The duality given by~\eqref{identityzdzo} has many important consequences.
First, because the normalized dessin partition is a tau-function of the KP hierarchy,
it follows from~\eqref{identityzdzo} that so is the normalized LUE partition function.
We will refer to it as the {\em LUE tau-function} of the KP hierarchy.
Since the affine coordinates of the dessin tau-function (of the KP hierarchy) are explicitly known,
so are the affine coordinates of the LUE tau-function.
Secondly, the LUE correlators have recently been shown to be related to
strictly monotone Hurwitz numbers and weakly monotone Hurwitz numbers
\cite{CDOC, GGR, GGN1, GGN2, GGN3}.
So a consequence of our result is a relationship between numbers of dessins
and %the number of 
strictly monotone Hurwitz numbers (see Corollary~\ref{cornklhurwitz}).
Thirdly, %by combining \eqref{eqn:Even-Des} with \eqref{identityzdzo},
one can recover \cite[Theorem~1.5]{GGR} (see \S~\ref{newsection3})  that relates strictly monotone Hurwitz numbers to
modified GUE partition function and hence to special triple Hodge integrals
by the Hodge-GUE correspondence \cite{DLYZ}.

The mystery of the appearance of the factor \eqref{correctionfactordesssins} dissolves
when one looks at it from the perspective of matrix model theory.
It is known from the orthogonal polynomial theory (cf.~e.g.~\cite{AvM,Deift,mehta}) that the LUE
model up to a constant gives a particular solution to the Toda lattice hierarchy, which we call the {\it LUE/dessin solution}.
The significance of the correction factor in~\eqref{correctionfactordesssins} is then clear in this context:
It is a normalization constant %from the matrix integral,
such that multiplying by this factor makes~$Z_{\rm dessins}$ a tau-function
of the LUE/dessin solution (cf.~\cite{CDZ,DY1,DZ} for the meaning of tau-function here); 
see Theorem~\ref{thm2}; a similar result for the GUE solution could be found in the Appendix of~\cite{DY1}.

However another  mysterious fact is that this factor is related to a factor  in Keating-Snaith Conjecture \cite{KS}
on the moments of Riemann zeta function.
Keating and Snaith proved that for $\Re(s) > -1$, the moments of power of the characteristic
polynomial in circular unitary ensemble (CUE) are given by the following formula:
\beq
\corr{|Z(U,\theta)|^s}_{U(N)}
= \prod_{j=1}^N \frac{\Gamma(j)\Gamma(j+s)}{(\Gamma(j+s/2))^2},
\eeq
where $Z(U,\theta)$ is the characteristic polynomial of the unitary matrix $U$:
\beq
Z(U,\theta) = \det \bigl(I-Ue^{-i\theta}\bigr) = \prod_{j=1}^N \bigl(1-e^{i(\theta_j-\theta)}\bigr),
\eeq
where $e^{i\theta_1}, \dots, e^{i\theta_N}$ are eiegenvalues of $U$.
From this they get:
\beq
f_{CUE}:=\lim_{N\to \infty} \frac{1}{N^{\lambda^2}} \corr{|Z(U,\theta)|^{2\lambda}}_{U(N)}
= \frac{(G(1+\lambda))^2}{G(1+2\lambda)}.
\eeq
Their conjecture is that for the moments of the Riemann zeta function \cite{KS, NY},
\beq
\lim_{T\to \infty} \frac{1}{(\log T)^{\lambda^2}} 
\frac{1}{T} \int_0^T \biggl|\zeta\big(\frac{1}{2}+it\big)\biggr|^{2\lambda}dt
= \frac{(G(1+\lambda))^2}{G(1+2\lambda)} A(\lambda),
\eeq
where $A(\lambda)$ is the arithmetic factor:
\beq
A(\lambda) = \prod_{p \; \text{prime}} \Biggl[(1-\frac{1}{p})^{\lambda^2}
\sum_{m=0}^\infty \biggl(\frac{\Gamma(\lambda+m)}{m!\Gamma(\lambda)}\biggr)^2p^{-m}\Biggr].
\eeq
Furthermore, Keating and Snaith proved the following formula:
\beq
\corr{|Z(U,\theta)|^te^{is \Img \log Z(U,\theta)}}_{U(N)}
= \prod_{j=1}^N \frac{\Gamma(j)\Gamma(j+t)}{\Gamma(j+(t+s)/2)\Gamma(j+(t-s)/2)},
\eeq
From this one gets \cite{BY}:
\beq \label{eqn:KS2}
\lim_{N\to \infty} \frac{1}{N^{(t^2-s^2)/4}} \corr{|Z(U,\theta)|^te^{is \Img \log Z(U,\theta)}}_{U(N)}
= \frac{G(1+\frac{t+s}{2})G(1+\frac{t-s}{2})}{G(1+t)}.
\eeq
By \eqref{eqn:G-Reflection}, the correction factor \eqref{correctionfactordesssins} can
be rewritten as:
\beq
\e^{-\frac1{12} +\frac{v^2}{\e^2}} (2 \pi)^{-\frac{(u+v)}{\e}}
 \exp \int_0^{\frac{u}\e} \pi t \cot (\pi t) dt \, 
\frac{G(1-\frac{u}\e)G(1+\frac{v}\e)}{G(1+\frac{v-u}\e)},
\eeq
which is almost the right-hand side of \eqref{eqn:KS2} with $t= \frac{v-u}{\e}$, $s= \frac{v+u}{\e}$.
Such coincidence suggests some connections between dessins and Riemann zeta function to be discovered in the future.

In fact the duality between dessin counting and LUE suggests another connection between
dessins and Riemann zeta function.
Note in physics literature, LUE is referred to as the generalized Penner model \cite{DV}.
For the reader's convenience
we recall some relevant facts about Penner model and the generalized Penner model in Appendix~\ref{appendixB}.
Originally, the Penner model was devised to prove the Harer-Zagier formula for the orbifold
Euler characteristics of $\cM_{g,n}$ by a better matrix model approach:
\beq
\chi(\cM_{g,n}) = (-1)^n \frac{(2g-3+n)!(2g-1)}{(2g)!} B_{2g}
= (-1)^{n-1} \frac{(2g-3+n)!}{(2g-2)!} \zeta(1-2g).
\eeq
The generalized Penner model, $c=1$ noncritical string theory and topological
string theory on the conifold are known in the physics literature to be related to each other,
and so our results suggest a connection of dessin counting with these theories. 
In viewpoint of Frobenius manifolds, there are deep connections among GUE, LUE, 
and the topological $\mathbb{P}^1$-sigma model. See the remarks at the end of Appendix~\ref{appendixB}.
We hope to explore such connections in future research.

Note that Toda lattice hierarchy and Ablowitz-Ladik hierarchy are two different reductions
of the 2-Toda hierarchy (cf.~e.g.~\cite{CY, DY1, UT}).
They are related by a transformation exchanging space and time variables~\cite{CDZ,Du0,V}.
The equivariant GW theory of the resolved conifold (local $\bP^1$) with anti-diagonal action is conjectured
by Brini~\cite{Brini} (see also~\cite{BCR, LYZZ}) to be governed by the Ablowitz-Ladik hierarchy.
We hope to investigate applications of such connections to 
generalize the results of this work 
 in the future.
 
\smallskip

The rest of the paper is organized as follows. In \S\,\ref{newsection2} we review on dessins, LUE and
Toda lattice. In \S\,\ref{newsection3} we prove Theorem~\ref{thm2} and Corollary~\ref{corzdzo}.  
In \S\,\ref{newsection4}, we apply the Dubrovin--Zhang method for
computing dessin correlators.
In \S\,\ref{newsection5} we give another proof of Corollary~\ref{corzdzo} and study its 
consequences. Concluding remarks are given in \S\,\ref{newsection6}.

\vspace{0.2in}
{\bf Acknowledgements}.
The second-named author is partly supported by NSFC grants 11661131005
and 11890662.

\section{Review on dessins, LUE, and one-dimensional Toda chain}\label{newsection2}
In this section, we review
Grothendieck's dessin counting, the normalized LUE partition function and the MR method~\cite{DY1,Y} to tau-functions for the Toda lattice hierarchy.

\subsection{Review on the enumeration of Grothendieck's dessins}\label{newsection21121}

In view of the definition~\eqref{deffdessins} it is obvious that
\beq\label{homogeneity3}
\sum_{j\geq1}  (j-1)  p_j  \frac{\p \F_{\rm dessins}}{\p p_j} - \e  \frac{\p \F_{\rm dessins}}{\p \e} - u  \frac{\p \F_{\rm dessins}}{\p u} -
v  \frac{\p \F_{\rm dessins}}{\p v} = 0
\eeq
and
\beq\label{inifdessins}
\F_{\rm dessins}(u,v,{\bf 0};\e)  \equiv 0,
\eeq
as well as that
\beq
\langle \tau_{\mu_1} \cdots \tau_{\mu_m} \rangle = \frac{\p \F_{\rm dessins}(u,v,{\bf p};\e)}{\p p_{\mu_1} \dots \p p_{\mu_m}}  \bigg|_{{\bf p}={\bf 0}} .
\eeq

According to~\cite{KZ}, the dessin partition function $Z_{\rm dessins}$
satisfies the following Virasoro constraints:
\beq\label{virazdessins}
L_b Z_{\rm dessins} = 0 , \quad b\geq 0,
\eeq
where
\beq\label{viraoperatorsdessins}
L_b = (u+v)  b  \frac{\p }{\p p_b} + \sum_{j\geq1}  %(p_j - \delta_{j,1}) 
(b+j) \tilde{p}_j \frac{\p }{\p p_{b+j}}
+ \frac1{\e^2}\sum_{i,j\ge1\atop i+j=b}  i  j  \frac{\p^2 }{\p p_i \p p_j}  + \frac{uv}{\e^2}  \delta_{b,0}
\eeq
are linear operators satisfying the Virasoro commutation relations:
\beq
[L_{b_1},L_{b_2}] = (b_1-b_2)  L_{b_1+b_2} , \quad \forall \, b_1,b_2\geq 0.
\eeq
Here, $\tilde p_j=p_j-\delta_{j,1}$ as in the Introduction.

The dilaton equation~\eqref{dilatonforZ121} written 
for the corrected dessin partition function~$Z$ follows from~\eqref{homogeneity3}, the $b=0$ case of~\eqref{virazdessins} and the 
definition~\eqref{correctionfactordesssins}.

It follows from the Virasoro constraints~\eqref{virazdessins} and the properties \eqref{homogeneity3}--\eqref{inifdessins} that
\beq\label{zdessinsew1}
Z_{\rm dessins} = e^{W} (1)
\eeq
(see~\cite{KZ,Zhou2}).
Here $W$ denotes the following cut-and-joint type linear operator:
\beq
W := (u+v)  \Lambda_1 + M_1 + \frac{uvp_1}{\e^2},
\eeq
where
\begin{align}
& \Lambda_1 = \sum_{i\geq2} (i-2)  p_i  \frac{\p }{\p p_{i-1}} , \\
& M_1 = \sum_{i\geq2}  \sum_{j=1}^{i-1}  \biggl((i-1)  p_j  p_{i-j}  \frac{\p}{\p p_{i-1}} + \e^2 j  (i-j)  p_{i+1}  \frac{\p^2}{\p p_j \p p_{i-j}}\biggr).
\end{align}
The following formula is proved in \cite{Zhou2}:
\beq
Z_{\rm dessins} = \sum_{\mu \in \cP} s_\mu  \prod_{\Box\in \mu} \frac{(u+c(\Box)) (v+c(\Box))}{h(\Box)},
\eeq
where the summation is taken over the set $\cP$ of all partitions,
and in the product
$\Box$ runs through the boxes of Young diagram of  the partition $\mu$,
and $c(\Box)$ and $h(\Box)$ denotes the content and hook length of $\Box$ respectively.
In the fermionic picture,
the dessin tau-function is given by a Bogoliubov transformation:
\beq \label{eqn:Z-A}
Z_{\rm dessins} = \exp\biggl(\sum_{m,n \geq 0} A_{m,n} \psi_{-m-\frac{1}{2}}\psi^*_{-n-\frac{1}{2}}\biggr) \vac,
\eeq
where the coefficients $A_{m,n}$ as explicitly given
as follows:
\beq \label{eqn:Amn}
A_{m, n}
= \frac{(-1)^n uv}{(m+n+1)m!n!}
\prod_{j=1}^{m} (u+j)(v+j) \prod_{i=1}^n (u-i)(v-i).
\eeq
Based on this we then get a formula for the $n$-point functions
associated with the dessin counting as stated in Theorem A by the general result in \cite{Zhou}.

\subsection{Review on LUE}
As in \S~\ref{section1}, denote by $\mathcal{H}_n^+$ the space of positive hermitian matrices of size~$n$.
Define {\it the normalized LUE partition function of size~$n$}~\cite{AvM, CDOC, GGR} by
\beq\label{defznlue1}
Z_n^{\rm LUE1}({\bf p};\alpha;\e) =
\frac{G(\alpha+1) \, \e^{-n^2-\alpha n}}{\pi^{\frac{n(n-1)}{2}}  G(n+\alpha+1)}
\int_{\mathcal{H}^+_n}  (\det{M})^\alpha  e^{- \frac1\e  {\rm tr} \, V(M; {\bf p}) }  dM,
\eeq
where $\alpha$ is a parameter,
$G(z)$ denotes the Barnes $G$-function, 
$V(M;{\bf p})$ is defined in~\eqref{vmdef}, and 
$dM$ is the invariant measure given by~\eqref{invmeasure}. 
Here we assume that ${\rm Re} (\alpha)>-1$. Later from the polynomiality of the connected LUE correlators we
can also view~$\alpha$ as an indeterminate.
The sum $n+\alpha=:w$ is called the {\it Wishart parameter},
as it is the size of the corresponding {\it Wishart matrix} (cf.~e.g.~\cite{CDOC}) for the case that $n+\alpha$ is a nonnegative integer.
For more details about the definition and literature see~e.g.~\cite{AvM, GGR}.
We note that, in the definition~\eqref{defznlue1}, the partition function $Z_n^{\rm LUE1}({\bf p};\alpha;\e)$
is understood as a power series of $p_1,p_2,\cdots$, namely, it should be
understood in the way that one first Taylor expands the integrand
$(\det{M})^\alpha  e^{- \frac1\e  {\rm tr} \, V(M;{\bf p}) }$ with respect to~${\bf p}$, and then
do the integration for the coefficient of each monomial of~${\bf p}$. Clearly, 
the normalized LUE partition function
$Z_{\rm LUE1}$ defined in~\eqref{defzlue1intro} in the Introduction relates to $Z_n^{\rm LUE1}$ by
\begin{align}
\label{defzonelue} Z_{\rm LUE1}(x, {\bf p}; a;\e) = Z_{x/\e}^{\rm LUE1}({\bf p};a/\e;\e).
\end{align}

It follows immediately from the definition~\eqref{defznlue1} that
\beq
\label{dilatonlue1}
\sum_{j\geq1} \tilde p_j \frac{\p Z_n^{\rm LUE1}({\bf p};\alpha;\e)}{\p p_j}
+ \e  \frac{\p Z_n^{\rm LUE1}({\bf p};\alpha;\e)}{\p \e}
%+ (n^2+\alpha n) Z_n^{\rm LUE1}({\bf p};\alpha;\e)
= 0 .
\eeq
Here and below, $\tilde p_j=p_j-\delta_{j,1}$. 
Let us also prove another property of $Z_n^{\rm LUE1}({\bf p};\alpha;\e)$ in the following lemma.
\begin{lemma}
The power series $Z_n^{\rm LUE1}({\bf p};\alpha;\e)$ satisfies that
\beq\label{lue111}
Z_n^{\rm LUE1}({\bf 0};\alpha;\e) \equiv 1.
\eeq
\end{lemma}
\begin{proof}
According to~\cite{GGR}, we know that
\beq
\int_{\mathcal{H}^+_n}  (\det{M})^\alpha  e^{-{\rm tr} M }  dM = \frac{\pi^{\frac{n(n-1)}{2}}  G(n+\alpha+1)}{G(\alpha+1)},
\eeq
which, after performing $M\to \e^{-1} M$ in the integration, implies formula~\eqref{lue111}.
\end{proof}
According to~\eqref{lue111}, the logarithm $\log Z_n^{\rm LUE1}({\bf p};\alpha;\e)$ is again a power series of~${\bf p}$.
For $m\geq1$, and $\mu_1,\dots,\mu_m\geq1$, the following derivative evaluated at ${\bf p}={\bf 0}$
\beq
\mu_1\cdots\mu_m \frac{\p^m\log Z_n^{\rm LUE1}({\bf p};\alpha;1)}{\p p_{\mu_1} \dots p_{\mu_m}}\bigg|_{{\bf p}={\bf 0}}
\eeq
is called a {\it connected LUE correlator}, often denoted by
$\langle \tr \, M^{\mu_1} \cdots \tr \, M^{\mu_m}\rangle_c$.

Let us now briefly recall the topological meaning of a connected LUE correlator obtained by
Cunden, Dahlqvist, and O'Connell~\cite{CDOC}.
Recall that a {\it partition} $\mu=(\mu_1,\mu_2,\dots)$ is a sequence of weakly decreasing
nonnegative integers with $\mu_k = 0$ for sufficiently large~$k$.
The {\it length} $\ell(\mu)$ is the number of the
nonzero parts of $\mu$, the {\it weight} $|\mu|:= \mu_1 + \mu_2 + \cdots$,
and $\mu$ is also called a partition of $|\mu|$.
Denote by $\mathcal{P}_d$ the set of all partitions of~$d$.
For $g,d$ being nonnegative integers and for $\mu,\nu\in\mathcal{P}_d$,
the strictly monotone double Hurwitz number $h_g(\mu,\nu)$ in genus~$g$ and degree~$d$
is defined as the number of tuples $(\alpha,\tau_1,\dots,\tau_r, \beta)$ with $r= \ell(\mu)+\ell(\nu)+2g-2$ satisfying
\begin{itemize}
\item $\alpha,\beta$ are permutations of $\{1,\dots,d\}$ of cycle type $\mu,\nu$, respectively, and $\tau_1,\dots,\tau_r$
are transpositions such that $\alpha\tau_1\cdots\tau_r=\beta$
\item the subgroup generated by $\alpha,\tau_1,\dots,\tau_r$ acts transitively on $\{1,\dots,d\}$
\item writing $\tau_j=(a_j,b_j)$ with $a_j<b_j$, $j=1,\dots,r$, then we have $b_1<\cdots<b_r$.
\end{itemize}
The following formula is proved in~\cite{CDOC}: For $\mu\in\mathcal{P}_d$, as $n\to\infty$,
\beq\label{CDOCexpansion}
n^{m-d-2} \langle \tr \, M^{\mu_1} \cdots \tr \, M^{\mu_{m}}\rangle_c
= \sum_{g\geq0} \frac1{n^{2g}} \sum_{s=1}^{1-2g+d-m} \frac{z_\mu}{|\mu|!} \sum_{\nu \in \mathcal{P}_d\atop \ell(\nu)=s}h_g(\mu;\nu) \, c^s, \quad c>1-\frac1n.
\eeq
Here $c=1+\frac{\alpha}n$, $z_\mu=\prod_{i\geq1} i^{n_i} n_i!$ with $n_i$ being the multiplicity of~$i$
in~$\mu$.

Before ending this section, let us point out that the above two properties~\eqref{dilatonlue1}--\eqref{lue111} for
$Z_n^{\rm LUE1}({\bf p};\alpha;\e)$ translate into those for~$Z_{\rm LUE1}(x, {\bf p}; a;\e)$ as follows:
\begin{align}
\label{dilatonlueone}
& \sum_{j\geq1}  \tilde p_j  \frac{\p Z_{\rm LUE1}}{\p p_j}
+ \e \frac{\p Z_{\rm LUE1}}{\p \e} + x \frac{\p Z_{\rm LUE1}}{\p x} + a \frac{\p Z_{\rm LUE1}}{\p a}
%+ \frac{x(x+a)}{\e^2} Z_{\rm LUE1} 
= 0 , \\
\label{zone1}
&Z_{\rm LUE1}(x; {\bf 0};a;\e) \equiv 1.
\end{align}
These two properties will be important for us in the next section.
We also define the {\it normalized LUE free energy} $\F_{\rm LUE1}(x, {\bf p}; a;\e)$ by
$\F_{\rm LUE1}(x, {\bf p}; a;\e)=\log Z_{\rm LUE1}(x; {\bf p};a;\e)$.
From~\eqref{zone1} we know that $\F_{\rm LUE1}(x, {\bf p}; a;\e)$ is a power series of~${\bf p}$.

\subsection{Review on the MR method to tau-functions for Toda lattice}
Let
\beq
L= \Lambda + V(x) + W(x) \Lambda^{-1}
\eeq
be a linear difference operator, called the {\it Lax operator}, where $\Lambda:f(x)\mapsto f(x+\e)$ is the shift operator. The
{\it Toda lattice hierarchy} is a system of evolutionary differential-difference equations, which can be defined by
\beq\label{TLH}
\frac{\p L}{\p t_i} = \frac1{\e}\bigl[(L^{i+1})_+,L\bigr], \quad i\geq0.
\eeq
Here, for a difference operator~$P$ written in the form $P=\sum_{k\in \ZZ} P_k \Lambda^k$,
$P_+$ is defined as $\sum_{k\geq0} P_k \Lambda^k$. The $\p_{t_0}$-flow reads explicitly as follows:
\begin{align}
& \frac{\p V}{\p t_0} = \frac1{\e} (W(x+\e)-W(x)), \quad \frac{\p W}{\p t_0} = \frac1{\e} W(x) (V(x)-V(x-\e)),
\end{align}
which is known as the {\it Toda equation}.
Denote by
\beq\mathcal{A}=\ZZ[V(x), W(x), V(x\pm \e), W(x\pm \e), V(x\pm 2\e), W(x\pm 2\e), \cdots]\eeq
the polynomial ring,
and denote
\beq
U(\lambda) = \begin{pmatrix} V(x)-\lambda & W(x) \\ -1 & 0 \\ \end{pmatrix}.
\eeq
Recall that~\cite{DY1} the {\it basic matrix resolvent} $R(\lambda)$ is defined as
the unique element in ${\rm Mat}(2,\mathcal{A}[[\lambda^{-1}]])$ satisfying
\begin{align}
& \Lambda(R(\lambda)) U(\lambda) - U(\lambda) R(\lambda) =0,\\
& \tr \, R(\lambda) =1, \quad \det R(\lambda)=0,\\
& R(\lambda) - \begin{pmatrix} 1 & 0 \\ 0 & 0 \\ \end{pmatrix} \in {\rm Mat}(2,\mathcal{A}[[\lambda^{-1}]]\lambda^{-1}). \label{entrymr421}
\end{align}
Write
\beq
R(\lambda) = \begin{pmatrix} 1+\alpha(\lambda) & \beta(\lambda) \\ \gamma(\lambda) & -\alpha(\lambda) \\ \end{pmatrix}.
\eeq
Let $S_i:={\rm Coef} \bigl(\Lambda(\gamma(\lambda)), \lambda^{-i-2}\bigr)$, $i\geq0$, and
define~\cite{DY1} a sequence of elements $\Omega_{i,j} \in \mathcal{A}$ by
\beq
\sum_{i,j\geq0} \frac{\Omega_{i,j}}{\lambda^{i+2} \mu^{j+2}} = \frac{\tr \, R(\lambda)R(\mu)}{(\lambda-\mu)^2} - \frac{1}{(\lambda-\mu)^2}.
\eeq
Then we have~\cite{DY1}
\begin{align}
&\Omega_{i,j} =\Omega_{j,i}, \quad \frac{\p \Omega_{i,j}}{\p t_r}=\frac{\p \Omega_{j,r}}{\p t_i}, \label{taustr1} \\
&(\Lambda-1)(\Omega_{i,j}) = \e\frac{\p S_i}{\p t_j}, \label{taustr2} \\
&(1-\Lambda^{-1}) (S_i) = \e\frac{\p \log W}{\p t_i}. \label{taustr3}
\end{align}
Therefore,
 %relations~\eqref{taustr1}--\eqref{taustr3} imply that
 for an arbitrary power-series-in-${\bf t}$ solution $(V(x,{\bf t};\e), W(x,{\bf t};\e))$ to the Toda lattice hierarchy,
there exist a power series $\tau(x,{\bf t};\e)$ in ${\bf t}=(t_0,t_1,t_2,\dots)$ such that
\begin{align}
& \e^2\frac{\p^2 \log\tau(x,{\bf t};\e)}{\p t_i \p t_j}=\Omega_{i,j}, \label{deftau1} \\
& \e\frac{\p}{\p t_i} \biggl(\log \frac{\tau(x+\e,{\bf t};\e)}{\tau(x,{\bf t};\e)}\biggr) = S_i, \label{deftau2}\\
& \frac{\tau(x+\e,{\bf t};\e)\tau(x-\e,{\bf t};\e)}{\tau(x,{\bf t};\e)^2} = W(x,{\bf t};\e). \label{deftau3}
\end{align}
We call $\tau(x,{\bf t};\e)$
the {\it tau-function} of the solution $(V(x,{\bf t};\e), W(x,{\bf t};\e))$ to
the Toda lattice hierarchy~\cite{DY1,DZ,Y}.

\begin{remark}
In the literature (see e.g.~\cite{AvM}), the requirement~\eqref{deftau3} is usually not specified and the 
correction factor in~\eqref{correctionfactordesssins} is not relevant in their consideration. But for us, this correction factor plays 
a crucial role because it is required in the Dubrovin--Zhang approach. 
\end{remark}

The following formula is proved in~\cite{DY1}: for each $m\geq2$,
\beq\label{multiderivativesresolvent}
\sum_{i_1,\dots,i_m\ge0}
\frac{\e^m\frac{\p^m \log \tau(x,{\bf t};\e)}{\p t_{i_1} \cdots \p t_{i_m}}}
{\lambda_1^{i_1+2}\cdots \lambda_m^{i_m+2}} = - \sum_{\sigma\in S_m/C_m}
\frac{\tr \, \prod_{j=1}^m R(\lambda_{\sigma(j)})}{\prod_{j=1}^m (\lambda_{\sigma(j)}-\lambda_{\sigma(j+1)})} - \frac{\delta_{m,2}}{(\lambda-\mu)^2}.
\eeq

Let us continue and review another formula~\cite{Y} for generating series for logarithmic derivatives of $\tau(x,{\bf t};\e)$.
Note that power-series-in-${\bf t}$ solutions $(V(x,{\bf t};\e),W(x,{\bf t};\e))$
to equations~\eqref{TLH} are in one-to-one correspondence with their initial values
\beq
V(x,{\bf 0};\e)=:f(x,\e), \quad W(x,{\bf 0};\e)=:g(x,\e).
\eeq

Let us denote
\beq
s(x;\e)= - (1-\Lambda^{-1})^{-1} (\log g(x;\e)).
\eeq
Recall that two elements
\beq \psi_A(\lambda) = (1+{\rm O}(\lambda^{-1}))\lambda^{x/\e}, \quad
\psi_B(\lambda)=(1+{\rm O}(\lambda^{-1}))e^{-s(x;\e)}\lambda^{-x/\e}\eeq
are called forming {\it a pair of wave functions}~\cite{Y} of the initial Lax operator
\beq
L_{\rm ini}=\Lambda+f(x;\e) + g(x;\e)\Lambda^{-1},
\eeq
if they satisfy the following conditions:
\begin{align}
& L_{\rm ini} (\psi_A) = \lambda \psi_A(\lambda), \quad L_{\rm ini} (\psi_B) = \lambda \psi_B(\lambda), \label{pairwave1}\\
& d(\lambda):= \psi_A(\lambda) \, \Lambda^{-1}(\psi_B(\lambda)) -  \psi_B(\lambda) \, \Lambda^{-1}(\psi_A(\lambda)) = \lambda e^{-s(x-\e;\e)}.\label{pairwave2}
\end{align}
The following formula is proved in~\cite{Y}: for each $m\geq2$,
\beq
\sum_{i_1,\dots,i_m\ge0}
\frac{\e^m\frac{\p^m \log \tau(x,t;\e)}{\p t_{i_1} \cdots \p t_{i_m}}\big|_{{\bf t}={\bf 0}} }
{\lambda_1^{i_1+2}\cdots \lambda_m^{i_m+2}} = (-1)^{m-1} \frac{e^{m s(x-\e;\e)}}{\prod_{j=1}^m \lambda_j} \sum_{\sigma\in S_m/C_m}
\prod_{j=1}^m D(\lambda_{\sigma(j)},\lambda_{\sigma(j+1)}) - \frac{\delta_{m,2}}{(\lambda-\mu)^2}, \label{Y421}
\eeq
where $D(\lambda,\mu)$ is defined by
\beq\label{dpsiab}
D(\lambda,\mu)= \frac{\psi_A(\lambda) \Lambda^{-1}(\psi_B(\mu)) - \Lambda^{-1}(\psi_A(\lambda)) \psi_B(\mu)}{\lambda-\mu}.
\eeq

\section{Proofs of Theorem~\ref{thm2} and Corollary~\ref{corzdzo}}\label{newsection3}
In this section, %based on 
we prove Theorem~\ref{thm2} and Corollary~\ref{corzdzo}, and then we introduce the corrected LUE partition function.

\begin{proof}[Proof of Theorem~\ref{thm2}] 
Consider the unique solution to the Toda lattice hierarchy~\eqref{TLH}
specified by the initial data~\eqref{initialvaluetodadessins}. Let $\tilde Z$ denote the tau-function 
of this solution. It suffices to show that $\log\tilde Z$ and $\log Z(x,{\bf p};a;\e)$ differ at most by an affine function of $x, {\bf p}$.
First of all, as a power series of~${\bf p}$, the constant term of~$\log \tilde Z$, i.e., 
$\log \tilde Z|_{{\bf p}={\bf 0}}$ should satisfy~\eqref{deftau3} at ${\bf p}=0$, and 
one can check that $\log Z(x,{\bf 0};a;\e)$ does satisfy this condition. Exponential of this constant term 
is the correction factor (cf.~\eqref{correctionfactordesssins} from the Introduction). So the constant terms of $\log \tilde Z$ and $\log Z(x,{\bf p};a;\e)$ agree.
Since $S_i$, $i\geq0$, and the coefficients of entries of $R(\lambda)$ are in the polynomial ring~$\mathcal{A}$, 
using~\eqref{deftau2}, \eqref{multiderivativesresolvent} and the initial values~\eqref{initialvaluetodadessins}, 
we know that $ \p_{p_j}(\Lambda-1)\log \tilde Z|_{{\bf p}={\bf 0}}$ as well as 
the logarithmic derivatives of~$\tilde Z$ of order higher than or equal to~2 with respect to~${\bf p}$ at ${\bf p}={\bf 0}$ 
have polynomial dependence in~$x$. It then suffices to show these derivatives of $\log\tilde Z$ and those of $\log Z$ agree. 
Without loss of generality we can now set $\epsilon=1$, namely,  
the initial Lax operator~$L_{\rm ini}$ is now given by 
\beq
L_{\rm ini} = \Lambda + (2n+\alpha+1) + n(n+\alpha)\Lambda^{-1}.
\eeq
\begin{lemma}
The elements
\begin{align}
& \psi_A(\lambda) := \lambda^{n} h(-\lambda,-n, -(n+\alpha)), \label{psiaexp}\\
& \psi_B(\lambda) := \lambda^{-n} \Gamma(1+n)\Gamma(1+n+\alpha) h(\lambda,n+1,n+\alpha+1). \label{psibexp}
\end{align}
form a pair of wave functions of~$L_{\rm ini}$. Here, 
\beq\label{defh}
h(\lambda,n,w):=\sum_{i\geq0} \frac1{i!}\frac{\prod_{r=0}^{i-1} \bigl((n+r)(w+r)\bigr)}{\lambda^i}= 1+\frac{n w}{\lambda}+\frac{nw (n+1) (w+1)}{2 \lambda ^2}+{\rm O}(\lambda^{-3}).
\eeq
\end{lemma}
\begin{proof}
Note that $s(n;1)$ for $L_{\rm ini}$ has the expression
\beq
s(n;1) = - \log\Gamma(1+n) - \log \Gamma(1+n+\alpha) .
\eeq
The statement is then
proved by verifying the conditions~\eqref{pairwave1}--\eqref{pairwave2}.
\end{proof}

Substituting~\eqref{psiaexp}--\eqref{psibexp} in~\eqref{dpsiab} we obtain
\begin{align}
e^{s(n-1;1)}  D(\lambda,\mu) = \frac1{\lambda-\mu} + \sum_{i,j\geq0} \frac{A_{i,j}}{\lambda^{i+1}\mu^{j+1}},
\end{align}
where $A_{i,j}$, $i,j\geq0$, are defined in~\eqref{defaij}. Therefore, by~\eqref{Y421} and~\eqref{connectedAformula} we find that the 
 derivatives of~$\log\tilde Z$ of order higher than or equal to~2 with respect to~${\bf p}$ at ${\bf p}={\bf 0}$ 
 coincide with those of $\log Z(n,{\bf p};a;1)$. 

Now recall that a product formula~\cite{Y} says that the entries of the basic matrix resolvent (see~\eqref{entrymr421}) can be expressed
by the pair of wave functions as follows:
\begin{align}
& \alpha(\lambda) = -1+\frac{\psi_A(\lambda) \Lambda^{-1}(\psi_B(\lambda))}{d(\lambda)} ,
\quad  \beta(\lambda)= - \frac{\psi_A(\lambda) \psi_B(\lambda)} {d(\lambda)}, \label{eq87}\\
& \gamma(\lambda) =  \frac{\Lambda^{-1}(\psi_A(\lambda)) \Lambda^{-1}(\psi_B(\lambda))}{d(\lambda)}, \label{eq88}
\end{align}
where we recall from~\eqref{pairwave2} that
$d(\lambda)= e^{-s(n-1;1)}\lambda$. So by~\eqref{deftau2},
\begin{align}
& j (\Lambda-1) \p_{p_{j}} \bigl(\log \tilde Z\bigr) |_{{\bf p}={\bf 0}} 
 = {\rm Coef} \bigl(\Lambda(\gamma(\lambda)), \lambda^{-j-1}\bigr) 
 = {\rm Coef} \biggl(\frac{\psi_A(\lambda) \psi_B(\lambda)}{\Lambda(d(\lambda))}, \lambda^{-j-1}\biggr) \label{eq89}\\
& \qquad =  {\rm Coef} \biggl(\frac{h(-\lambda,-n, -(n+\alpha)) h(\lambda,n+1,n+\alpha+1)}
{\lambda}, \lambda^{-j-1}\biggr). \nn
\end{align}
Here $j\ge1$. On the other hand, from~\eqref{virazfull} with $k=0$ we know that
\beq
 \sum_{j\geq1}  j p_j   \frac{\p \F}{\p p_{j}}
 + \frac{x(x+a)}{\e^2}  = \frac{\p \F}{\p p_1}.
\eeq
Differentiating both sides with respect to $p_j$ and then taking ${\bf p}={\bf 0}$ we obtain
\beq
 j \frac{\p \F}{\p p_j} \Big|_{{\bf p}={\bf 0}}
 = \frac{\p^2 \F}{\p p_1 \p p_j}\Big|_{{\bf p}={\bf 0}} = \frac{\p^2 \log \tilde Z}{\p p_1 \p p_j}\Big|_{{\bf p}={\bf 0}}.
\eeq
Then by using~\eqref{eq87} and the equation~(56) of~\cite{Y} we find that 
\beq\label{eq92}
j^2\langle\tau_{j}\rangle|_{u=n, v=n+\alpha,\e=1} = {\rm Coef} \bigl(h(-\lambda,-n,-(n+\alpha)) h(\lambda,n,n+\alpha), \lambda^{-j-1}\bigr),
\eeq
where $j\ge1$. Applying $\Lambda-1$ on both sides of~\eqref{eq92} and comparing it with~\eqref{eq89}, 
we find that $(\Lambda-1) \p_{p_{j}} \bigl(\log \tilde Z\bigr) |_{{\bf p}={\bf 0}}=(\Lambda-1)\langle\tau_{j}\rangle|_{u=n,v=n+\alpha,\e=1}$.
The theorem is proved.
\end{proof}

Let us continue and prove the following proposition.

\begin{prop}\label{thm1}
For each $m\geq2$, we have
\beq\label{dessinsformulaM}
\sum_{\mu_1,\dots,\mu_m\ge1}  \prod_{j=1}^{m} \mu_j 
\frac{\langle\tau_{\mu_1}\cdots\tau_{\mu_m}\rangle(n,w,1)}
{\lambda_1^{\mu_1+1}\cdots \lambda_m^{\mu_m+1}} = - \sum_{\sigma\in S_m/C_m}
\frac{\tr \, \prod_{j=1}^m M(\lambda_{\sigma(j)})}{\prod_{j=1}^m (\lambda_{\sigma(j)}-\lambda_{\sigma(j+1)})} - \frac{\delta_{m,2}}{(\lambda-\mu)^2},
\eeq
where $M(\lambda)$ is a two by two matrix given explicitly by
\beq\label{matrixresolventlue}
M(\lambda)=
\begin{pmatrix}
h(-\lambda,-n,-w) h(\lambda,n,w)  & - \frac{nw}{\lambda} h(-\lambda,-n,-w) h(\lambda,1+n,1+w) \\
\frac1{\lambda} h(-\lambda,1-n,1-w) h(\lambda,n,w)& - \frac{nw}{\lambda^2} h(-\lambda,1-n,1-w) h(\lambda,1+n,1+w)\\
\end{pmatrix}.
\eeq
Moreover, for $m=1$, we have
\beq\label{onepointlue}
1+\sum_{\mu\ge1}
\frac{\mu^2\langle\tau_{\mu}\rangle|_{\e=1}}
{\lambda^{\mu+1}} = h(-\lambda,-n,-w) h(\lambda,n,w).
\eeq
\end{prop}
\begin{proof}
Formula~\eqref{dessinsformulaM} follows from~\eqref{multiderivativesresolvent} evaluated at ${\bf p}=0$, with $M(\lambda)$ being
the basic resolvent at ${\bf p}=0$ whose expression can be found by using~\eqref{psiaexp}--\eqref{psibexp} and \eqref{eq87}--\eqref{eq88}.
Formula~\eqref{onepointlue} was already given by~\eqref{eq92}. The proposition is proved.
\end{proof}

Formula~\eqref{eq92}, or say~\eqref{onepointlue}, is equivalent to known formulas in e.g.~\cite{Zhou2}. Formulas~\eqref{dessinsformulaM}--\eqref{onepointlue} lead to an efficient algorithm of computing the dessin correlators.

Based on Proposition~\ref{thm1} and a result of M.~Gissoni, T.~Grava and G.~Ruzza~\cite{GGR},
 let us give a proof of Corollary~\ref{corzdzo}.

\begin{proof}[Proof of Corollary~\ref{corzdzo}] 
The following explicit formula for the generating series for the connected LUE correlators 
$\langle \tr \, M^{\mu_1} \cdots \tr \, M^{\mu_m}\rangle_c$ 
was derived in~\cite{GGR}:
\beq\label{LUEGGRM}
\sum_{\mu_1,\dots,\mu_m\ge1} 
\frac{\langle \tr \, M^{\mu_1} \cdots \tr \, M^{\mu_m}\rangle_c}
{\lambda_1^{\mu_1+1}\cdots \lambda_m^{\mu_m+1}} = - \sum_{\sigma\in S_m/C_m}
\frac{\tr \, \prod_{j=1}^m R_{\rm GGR}(\lambda_{\sigma(j)})}{\prod_{j=1}^m (\lambda_{\sigma(j)}-\lambda_{\sigma(j+1)})} - \frac{\delta_{m,2}}{(\lambda-\mu)^2},
\eeq
where $R_{\rm GGR}(\lambda)$ is a two by two matrix given explicitly by
\beq\label{matrixresolventlueGGR}
R_{\rm GGR}(\lambda)=\begin{pmatrix} 1 & 0 \\  0 & 0 \end{pmatrix} + 
\begin{pmatrix}
\ell A_\ell(n,n+\alpha)  & B_\ell(n+1,n+\alpha+1) \\
-n(n+\alpha)B_{\ell}(n,n+\alpha) & - \ell A_\ell(n,n+\alpha) \\
\end{pmatrix}
\eeq
with 
\begin{align}
& A_\ell(n,n') = \left\{ 
\begin{array}{ll} 
n, & \ell=0, \\ 
\frac1{\ell} \sum_{j=0}^{\ell-1} (-1)^j \frac{(n-j)_\ell (n'-j)_{\ell}}{j! (\ell-1-j)!}, & \ell\geq 1, \\
\end{array} \right. \\ 
& B_{\ell}(n,n') = \sum_{j=0}^\ell (-1)^j \frac{(n-j)_\ell (n'-j)_\ell}{j!(\ell-j)!} .
\end{align}
By a straightforward calculation we observe that $R_{\rm GGR}(\lambda)$ is related to $M(\lambda)$ by 
\beq\label{TRTiM}
T  R_{\rm GGR}(\lambda) T^{-1} \equiv M(\lambda),
\eeq
with $w=n+\alpha$, where  
\beq
T := \begin{pmatrix}
1 & 0 \\
0& - \frac{n}{n+\alpha}
\end{pmatrix}.
\eeq
From~\eqref{LUEGGRM}, \eqref{TRTiM} and~\eqref{dessinsformulaM} we know that 
$\log Z_{\rm dessins}(x,x+a,{\bf p};\e)$ and $\log Z_{\rm LUE1}(x,{\bf p};a;\e)$ can only 
differ by an affine function of ${\bf p}$ (coefficients of this affine function can depend on $x$, $\epsilon$). 
The identification between the connected one-point correlators of dessins and of LUE 
can be given straightforwardly as they are derived explicitly in~\eqref{onepointlue} and in~\cite{GGR}, respectively. 
The corollary is proved by observing that the constant terms of $\log Z_{\rm dessins}(x,x+a,{\bf p};\e)$ 
and of $\log Z_{\rm LUE1}(x,{\bf p};a;\e)$ (viewed as power series of ${\bf p}$) are 
both normalized as zero (see~\eqref{deffdessins} and \eqref{lue111}, respectively).
\end{proof}

We note that the method used in~\cite{GGR} for deriving explicit formula for 
the generating series of connected LUE correlators~\eqref{LUEGGRM} is different from ours: 
In~\cite{GGR, GGR2},  
isomonodromic tau-functions of certain Riemann-Hilbert problems are used.

Based on Corollary~\ref{corzdzo}, let us translate the correction factor~\eqref{correctionfactordesssins} to
the LUE partition function, namely, define the {\it corrected LUE partition function of size~$n$} as follows:
\begin{align}
Z_n^{\rm LUE}({\bf p};\alpha;\e) := & \, \e^{-\frac1{12} + n(n+\alpha)} (2 \pi)^{-n}
\frac{G(n+1) G(n+\alpha+1)}{G(\alpha+1)} Z_n^{\rm LUE1}({\bf p};\alpha;\e), \label{defluecorrect}\\
= & \, 
\frac{G(n+1) \e^{-\frac1{12}}}{ \pi^{\frac{n(n+1)}{2}} 2^n }
\int_{\mathcal{H}^+_n} (\det{M})^\alpha e^{-\frac1\e {\rm tr} \, V(M;{\bf p})} dM. \label{defznlue} 
\end{align}
For the special case when $\alpha=-1/2$, 
%motived from the Hodge-GUE correspondence~\cite{DLYZ, DY2},  
the corrected LUE partition function was also defined in~\cite{GGR}.
%Similarly as in the derivation of~\eqref{dilatonlue1}, 
It follows from~\eqref{defznlue} that (cf.~\eqref{dilatonlue1})
\beq
\sum_{j\geq1}  \tilde p_j  \frac{\p Z_n^{\rm LUE}({\bf p};\alpha;\e)}{\p p_j}
+ \e \frac{\p Z_n^{\rm LUE}({\bf p};\alpha;\e)}{\p \e}
+\frac1{12} Z_n^{\rm LUE}({\bf p};\alpha;\e)
= 0 .
\eeq

From the above translation and 
Corollary~\ref{corzdzo}, we have 
\beq\label{LUEpartitionfunctionequalZ}
Z(x,{\bf p}; a; \e) = Z_{x/\e}^{\rm LUE}({\bf p};a/\e;\e).
\eeq
So we also call the corrected dessin partition function $Z(x,{\bf p}; a; \e)$ {\it the corrected LUE partition function}
and call the corrected dessin free energy $\F(x, {\bf p}; a;\e)$ the {\it corrected LUE free energy}.

More consequences of Corollary~\ref{corzdzo} will be given in \S~\ref{newsection5}. 

It is convenient to denote 
\beq
\F(x, {\bf p}; a;\e)=:\e^{2g-2} \F_g(x,{\bf p};a).
\eeq
We call $\F_g(x,{\bf p};a)=\F_g$ the genus~$g$ part of the corrected dessin/LUE free energy, 
or for short the the genus~$g$ corrected dessin/LUE free energy.
In the next section, we will provide an application of Theorem~\ref{thm2} on the computation of~$\F_g$.

\section{Calculating dessin correlators via $\mathbb{P}^1$ Frobenius manifold}\label{newsection4}

In this section, following~\cite{Du0,Dubrovin,DY1,DY2,DZ-norm, DZ} we outline the Dubrovin--Zhang
method for calculating the corrected dessin free energy~$\F$. This method was briefly outlined for
the computation of the GUE free energy in~\cite{Dubrovin,DY1}, but here we will give more details.

Start with the Frobenius manifold~$M$ associated with the $\bP^1$-topological sigma model with the invariant flat metric~$\eta$~\cite{Du1}.
Take ${\bf v}=(v^1,v^2)$ a system of flat coordinates associated with~$\eta$, normalized by requiring that
$\p/\p v^1$ is the unit vector field and that $\eta_{\alpha\beta}=\delta_{\alpha+\beta,3}$.
We will sometimes write $v=v^1, u=v^2$ for simplification of notations.
Denote $(\eta^{\alpha\beta})=(\eta_{\alpha\beta})^{-1}$. (Avoid from confusing with the indeterminates $u,v$ in~\eqref{defdessinscorr}.)
Here and below, free Greek indices take the integer values $1,2$, the Einstein summation convention
is used for repeated Greek indices with one-up and one-down, and we use
$\eta_{\alpha\beta}$ and
 $\eta^{\alpha\beta}$ to lower and raise the Greek indices, respectively.
The Frobenius potential~\cite{Du0} of~$M$ is given by~\eqref{todafrob}
with the Euler vector field~$E$ on~$M$ given by~\eqref{todaeuler}.
In~\cite{Du0,DZ-norm,DZ} Dubrovin and Zhang proposed
a way of computing Gromov--Witten (GW) invariants of $\mathbb{P}^1$
via geometry of this Frobenius manifold and proved that the
partition function of GW invariants of $\mathbb{P}^1$
is a particular tau-function of a certain extension of the Toda lattice hierarchy (see
also~\cite{CDZ, Du0, EY, Getzler, OP, Zhang}). In other words, the
Dubrovin--Zhang hierarchy for the Frobenius manifold of $\mathbb{P}^1$
is normal Miura equivalent~\cite{DZ-norm} to the extended Toda hierarchy.
In particular, the corresponding principal hierarchy~\cite{Du0} coincides with the dispersionless
extended Toda hierarchy, which can be written as follows:
\begin{align}\label{phdef}
\frac{\p v^\alpha}{\p T^{\beta,q}} = \eta^{\alpha\gamma}\p_{x} \biggl(\frac{\p \theta_{\beta,q+1}({\bf v})}{\p v^\gamma}\biggr), \quad q\geq0.
\end{align}
Here $\theta_{\alpha,p}({\bf v})$, $p\ge0$, are certain smooth functions of~${\bf v}$
which can be defined via the generating series~\cite{DZ}
\begin{align}
&\theta_1({\bf v};z):=\sum_{p\geq0} \theta_{1,p}({\bf v}) z^p
= -2 e^{z v} \sum_{m\geq0} \Bigl(\gamma-\frac12 u+\psi(m+1)\Bigr) e^{mu} \frac{z^{2m}}{m!^2}, \label{theta1z}\\
&\theta_2({\bf v};z):=\sum_{p\geq0} \theta_{2,p}({\bf v}) z^p = z^{-1} \biggl(\sum_{m\geq0} e^{m u+z v} \frac{z^{2m}}{m!^2}-1\biggr), \label{theta2z}
\end{align}
where $\gamma$ is the Euler--Mascheroni constant and $\psi(z)$ is the digamma function.
The $\p_{T^{1,0}}$-flow reads:
\beq
\frac{\p v^\alpha}{\p T^{1,0}} = \frac{\p v^\alpha}{\p x}.
\eeq
Therefore, we identify $T^{1,0}$ with~$x$.
The $\p_{T^{2,p}}$-flows ($p\geq0$), also known as the {\it stationary flows},
are identified with the dispersionless Toda lattice hierarchy under the rescalings
\beq
(p+1)! \, \p_{T^{2,p}} = \p_{t_p}, \quad p\geq0.
\eeq
For the purpose of computing the dessin correlators we mostly look at the $\p_{T^{2,p}}$-flows ($p\geq0$), and 
for simplification of notations we denote ${\bf T}=(T^{2,p})_{p\geq0}$.

Our temporary goal is to find the power-series-in-${\bf T}$ solution to the $\p_{T^{2,p}}$-flows in the principal
hierarchy~\eqref{phdef} with
the following initial condition:
\beq\label{initialdata}
v(x,{{\bf T}={\bf 0}};a) = 2x+a, \quad e^{u(x,{{\bf T}={\bf 0}};a)} = x(x+a).
\eeq
Here $a$ is a parameter.
We know that power-series-in-${\bf T}$
solution $v^\alpha(x,{\bf T};a)$ to this initial value problem \eqref{phdef}, \eqref{initialdata} exists and is unique,
as~\eqref{phdef} is an {\it integrable system of evolutionary PDEs}.
This unique solution belongs to the class of {\it monotone solutions} and could be obtained by the
hodograph method~\cite{Du1,DGKM,DZ,Pavlov,Tsarev}. Indeed,
the monotonicity can be seen, via 
a direct computation of the Riemann invariants
\beq
R_1({\bf v})=v+2e^{u/2}, \quad R_2({\bf v})=v-2e^{u/2},
\eeq
 by the fact that $(R_i)_x$ ($i=1,2$) do not vanish for this unique solution at generic~$x=x_0$ and at ${\bf T}={\bf 0}$.
The hodograph method then shows that this unique solution coincides with the unique power-series-in-${\bf T}$ solution
to the following equation
\beq\label{displessEL}
x\frac{\p \theta_{1,0}({\bf v})}{\p v^\alpha} +\sum_{q\geq0} \widetilde{T}^{2,q} \frac{\p \theta_{2,q}({\bf v})}{\p v^\alpha} = \frac{\p \phi({\bf v};a)}{\p v^\alpha},
\eeq
where $\widetilde{T}^{2,q} =T^{2,q}-\delta_{q,0}$,
and $\phi({\bf v};a)$ is defined by
\beq\label{defphi}
\phi({\bf v};a):=-\frac{a}2u-\frac{a}2 \, {\rm log}\biggl(\frac{v+\sqrt{\Theta({\bf v})}}{v-\sqrt{\Theta({\bf v})}}\biggr) ,\quad \Theta({\bf v}):=v^2-4 e^{u}.
\eeq

Following~\cite{Du1},
define a family of smooth functions $\Omega_{\alpha,p;\beta,q}^{[0]}({\bf v})$, $p,q\geq0$, via the generating series:
\beq\label{defiOmegapq}
\sum_{p,q\geq0} \Omega_{\alpha,p;\beta,q}^{[0]}({\bf v}) z_1^p z_2^q =
\frac{\frac{\p \theta_\alpha({\bf v};z_1)}{\p v^\rho}\eta^{\rho\sigma}\frac{\p \theta_\beta({\bf v};z_2)}{\p v^\sigma}-\eta_{\alpha\beta}}{z_1+z_2},
\eeq
where $\theta_\alpha({\bf v};z)$ are given in~\eqref{theta1z}--\eqref{theta2z}.
These smooth functions are called the {\it genus zero two-point correlation functions}
of the $\mathbb{P}^1$ Frobenius manifold and
have the properties:
\begin{align}
& \Omega_{\alpha,p;\beta,q}^{[0]}({\bf v})=\Omega_{\beta,q;\alpha,p}^{[0]}({\bf v}), \quad
\p_{T^{\gamma,s}}(\Omega_{\alpha,p;\beta,q}^{[0]}({\bf v}))=\p_{T^{\beta,q}}(\Omega_{\alpha,p;\gamma,s}^{[0]}({\bf v})),
\quad \forall\,p,q,s\geq0, \label{omegapro1}\\
& \theta_{\alpha,p}({\bf v})=\Omega_{\alpha,p;1,0}({\bf v}), \quad \forall\, p\geq0, \label{omegapro2} \\
& E(\theta_{2,p}({\bf v}))=(p+1) \theta_{2,p}({\bf v}), \quad
E(\Omega_{2,p;2,q}({\bf v}))=(p+q+2) \Omega_{2,p;2,q}({\bf v}), \quad \forall\, p,q\geq0. \label{homogEtheta2pOm22}
\end{align}
Let us also consider a Hamiltonian PDE
\beq\label{yflow}
\frac{\p v^\alpha}{\p y} = \eta^{\alpha\gamma}\p_{x} \biggl(\frac{\p h({\bf v};a)}{\p v^\gamma}\biggr),
\eeq
with Hamiltonian density~$h=h({\bf v};a)$ given by
\beq
h({\bf v};a) := -\frac{a}2 uv+ a \sqrt{\Theta({\bf v})} - \frac{a}2 \, v \, {\rm log} \Biggl(\frac{v+\sqrt{\Theta({\bf v})}}{v-\sqrt{\Theta({\bf v})}}\Biggr).
\eeq
It is easy to verify that
\beq
h_{uu} = e^{u}h_{vv}.
\eeq
So the $\p_y$-flow commutes~\cite{DZ-norm, DZ} with the principal hierarchy~\eqref{phdef}. We now define
\begin{align}
&\omega({\bf v};a) = 
a^2 \log\Biggl(\frac{v+\sqrt{\Theta({\bf v})}}{2\sqrt{\Theta({\bf v})}}\Biggr), \label{defiomega}\\
&g_{p}({\bf v};a) = - \frac{a}2 \, \theta_{2,p}({\bf v}) + \frac{a}2 \, \frac{\sqrt{\Theta({\bf v})}}{(p+1)!} \sum_{k=0}^{p/2} \binom{p}{2k}\binom{2k}{k} e^{k u} v^{p-2k}. \label{defigp}
\end{align}
It is straightforward to verify that these functions including~$\phi({\bf v};a)$ satisfy the relations:
\begin{align}
& \p_x(\omega({\bf v};a)) = \p_y(\phi({\bf v};a)), \quad \p_{T^{2,p}}(\omega({\bf v};a)) = \p_y(g_p({\bf v};a)), \quad \forall\,p\ge0,\\
& \p_x(g_p({\bf v};a)) = \p_y(\theta_{2,p}({\bf v};a)), \quad \p_{T{2,q}}(g_p({\bf v};a)) = \p_y(\Omega_{2,p;2,q}({\bf v})), \quad \forall\, p,q\geq0,\\
& E(\phi({\bf v};a))=-a, \quad E(g_p({\bf v};a))=(p+1) g_p({\bf v};a), \quad E(\omega({\bf v};a))=0, \quad \forall\,p\geq0.\label{homogEphifmgpomega}
\end{align}

Following~\cite{Du0,DZ-norm,DZ}, introduce a power series $\widetilde{\F}_0=\widetilde{\F}_0(x,{\bf T};a)$ given by
\begin{align}
& \widetilde{\F}_0(x,{\bf T};a)
 :=  \frac12\sum_{p,q\geq0} \widetilde T^{2,p} \widetilde T^{2,q} \Omega_{2,p;2,q}^{[0]}({\bf v}(x,{\bf T};a))
+ \sum_{p\geq0} x \widetilde T^{2,p}\theta_{2,p}({\bf v}(x,{\bf T};a)) \label{dessins0dub} \\
& + \frac12 x^2 u(x,{\bf T};a) + \frac12 \omega({\bf v}(x,{\bf T};a);a) - x \phi({\bf v}(x,{\bf T};a);a)
 - \sum_{p\geq0} \widetilde T^{2,p} g_{p}({\bf v}(x,{\bf T};a);a), \nn
\end{align}
where ${\bf v}(x,{\bf T};a)=(v(x,{\bf T};a),u(x,{\bf T};a))$ denotes the power-series-in-${\bf T}$ solution to the
$\p_{T^{2,p}}$-flows ($p\geq0$) in the principal hierarchy~\eqref{phdef} specified by~\eqref{initialdata}. We are to show that
$\widetilde{\F}_0(x,{\bf T};a)$ coincides with the genus zero part $\F_0(x,{\bf p};a)$ of the corrected dessin/LUE free energy. First of all,
the following relations are known:
\begin{align}
& \frac{\p^2 \widetilde{\F}_0(x,{\bf T};a)}{\p T^{2,p} \p T^{2,q}} 
= \Omega_{2,p;2,q}^{[0]}({\bf v}(x,{\bf T};a)), \quad \forall \, p,q\geq0, \label{defgenus0tau1}\\
& \frac{\p^2 \widetilde{\F}_0(x,{\bf T};a)}{\p T^{1,0} \p T^{1,0}} = u(x,{\bf T};a), \label{defgenus0tau2}\\
& \frac{\p^2 \widetilde{\F}_0(x,{\bf T};a)}{\p T^{1,0} \p T^{2,p}} = \Omega_{1,0;2,p}^{[0]}({\bf v}(x,{\bf T};a)), \quad \forall \, p\geq0  \label{defgenus0tau3}
\end{align}
(cf.~\cite{Du0, DZ-norm, DZ}). It follows from these relations that~$e^{\e^{-2} \widetilde{\F}_0(x,{\bf T};a)}$ satisfies the genus zero part of the
defining equations~\eqref{deftau1}--\eqref{deftau3} for the tau-function of the solution, in other words, 
$e^{\e^{-2} \widetilde{\F}_0(x,{\bf T};a)}$ is the tau-function of the solution ${\bf v}(x,{\bf T};a)$ to the dispersionless Toda lattice hierarchy (cf.~\eqref{phdef}).
Here, we recall that 
%the independent variables~${\bf T}$ and ${\bf p}$ are related by
\beq\label{Tpt}
T^{2,q} = q! \, p_{q+1} = (q+1)! \, t_q,  \quad q\geq0.
\eeq

We note that the above relations \eqref{defgenus0tau1}--\eqref{defgenus0tau3}
can uniquely determine $\widetilde{\F}_0$ up to a linear function of $x, {\bf T}$ (recall that $x=T^{1,0}$).
Let us continue to show that the power series~$\widetilde{\F}_0$ defined by~\eqref{dessins0dub}
also satisfies the genus zero part of the $k=0$ case of equations~\eqref{virazfull}. To this end,
introduce the derivation
\beq
\mathcal{E}=\sum_{p\ge0} (p+1) \widetilde T^{2,p} \frac{\p }{\p T^{2,p}}.
\eeq
Using the first relation in~\eqref{homogEtheta2pOm22} and using the first relation in~\eqref{homogEphifmgpomega},
as well as using~\eqref{displessEL}, we find that
\beq
\sum_{p\geq0} (p+1)\widetilde T^{2,p} \theta_{2,p}({\bf v}(x,{\bf T};a))=\sum_{p\geq0} \widetilde T^{2,p}
E(\theta_{\alpha,p}({\bf v}))(x,{\bf T};a)  = -2x-a.
\eeq
Differentiating this equation with respect to $T^{1,0}$ and $T^{2,0}$, respectively, we get
\beq\label{homogMEuv}
\mathcal{E}(u(x,{\bf T};a))=-2, \quad \mathcal{E}(v(x,{\bf T};a))= -v(x,{\bf T};a)\,.
\eeq
By using the definitions \eqref{defiOmegapq}, \eqref{defiomega}, \eqref{defigp},
as well as using~\eqref{homogEphifmgpomega},
we then obtain
\begin{align}
& \mathcal{E}(\Omega^{[0]}_{2,p;2,q}({\bf v}(x,{\bf T};a))) = 
-(p+q+2) \Omega^{[0]}_{2,p;2,q}({\bf v}(x,{\bf T};a)), \quad \forall\, p,q\ge0,\label{homogOmega} \\
& \mathcal{E}(g_p({\bf v}(x,{\bf T};a);a))= -(p+1) g_p({\bf v}(x,{\bf T};a);a), \quad \forall\, p\geq0, \label{homogEgp}\\
& \mathcal{E}(\phi({\bf v}(x,{\bf T};a);a))= a, \quad \mathcal{E}(\omega({\bf v}(x,{\bf T};a);a))=0. \label{homogEphoom}
\end{align}
Hence by using \eqref{homogMEuv}--\eqref{homogEphoom} and~\eqref{dessins0dub} we find that
\beq\label{genuszeroEF0}
 \sum_{p\geq0} (p+1) \widetilde T^{2,p}  \frac{\p \widetilde{\F}_0(x,{\bf T};a)}{\p T^{2,p}} + x(x+a)  = 0.
\eeq
We conclude from \eqref{defgenus0tau1}--\eqref{defgenus0tau3} and~\eqref{genuszeroEF0}
that $\widetilde{\F}_0(x,{\bf T};a)$ could differ from~$\F_0(x,{\bf p};a)$ only possibly by adding a function
of~$x,a$ (with at most linear dependence in~$x$).
The final verification is a straightforward computation by taking ${\bf T}={\bf 0}$ in~\eqref{dessins0dub}
and by comparing the resulting expression $\widetilde{\F}_0|_{{\bf T}={\bf 0}}$ with
\beq
\F_0|_{{\bf p}={\bf 0}}=\frac{x^2}2 \log x + \frac{(x+a)^2}{2}  \log (x+a)
- \frac{a^2}2 \log a-\frac{3}{2} x(x+a).
\eeq
In a similar way, we can also verify that 
\beq\label{dilatonF0}
 \sum_{p\geq0} \widetilde T^{2,p} \frac{\p \widetilde{\F}_0(x,{\bf T};a)}{\p T^{2,p}} + x \frac{\p \widetilde{\F}_0(x,{\bf T};a)}{\p x} + 
 a \frac{\p \widetilde{\F}_0(x,{\bf T};a)}{\p a} = 2 \widetilde{\F}_0(x,{\bf T};a).
\eeq

Let us proceed with the higher genera. Recall that, in~\cite{DZ-norm,DZ}
the {\it quasi-triviality} for the extended Toda hierarchy~\cite{CDZ} is obtained.
%, which
%transforms the principal hierarchy to the extended Toda hierarchy.
It has the following form
\begin{align}
& \widehat V = \frac{\Lambda-1}{\e\p_x} (v) + (\Lambda-1) \circ \p_{T^{2,0}}  \Biggl(\sum_{g\geq1} \e^{2g-1} F^M_g\biggl({\bf v}, \frac{\p {\bf v}}{\p x},\dots,\frac{\p^{3g-2} {\bf v}}{\p x^{3g-2}}\biggr)\Biggr), \label{qm1}\\
& \widehat W = \frac{(\Lambda+\Lambda^{-1}-2)} {\e^2 \p_x^2} (u) + (\Lambda+\Lambda^{-1}-2) \Biggl(
\sum_{g\geq1} \e^{2g-2} F^M_g\biggl({\bf v}, \frac{\p {\bf v}}{\p x},\dots,\frac{\p^{3g-2}{\bf v}}{\p x^{3g-2}}\biggr)\Biggr). \label{qm2}
\end{align}
Here, $\Lambda=e^{\e \p_x}$, and $F^M_g=F^M_g({\bf v}_0={\bf v}, {\bf v}_1,\dots, {\bf v}_{3g-2})$ ($g\geq1$), with ${\bf v}_k=(v_k,u_k)$,
are absolute functions satisfying
\begin{align}
& \sum_{k\geq1} k u_k \frac{\p F^M_g}{\p u_k} + \sum_{k\geq1} k v_k \frac{\p F^M_g}{\p v_k} 
= (2g-2) F_g^M + \frac{1}{12} \delta_{g,1},\label{homogenfg}
\end{align}
and they have the form:
for $g\ge2$, $F^M_g$ are polynomials of ${\bf v}_2,{\bf v}_3,\dots,{\bf v}_{3g-2}$ and have
rational dependence in ${\bf v}_0,{\bf v}_1$; for $g=1$,
\begin{align}\label{f1gw}
F^M_1= \frac1{24} \log \bigl((v_1)^2-e^{u} (u_1)^2\bigr) - \frac1{24} u;
\end{align}
for $g\geq2$, the rational functions $F^M_g$ are uniquely determined by the {\it loop equation}~\cite{DZ-norm, DZ}:
\begin{align}
\label{loop}
&\quad\ \sum_{r\geq 0}\biggl({\p \Delta F^M\over \p v_r}
\biggl( \frac{v-\lambda} D \biggr)_r - 2 {\p \Delta F^M \over \p u_r} \biggl(\frac1 D\biggr)_r\biggr)\\
&\quad +\sum_{r\geq 1} \sum_{k=1}^r \binom{r}{k}
\biggl({1\over \sqrt{D}}\biggr)_{k-1} \biggl[ {\p \Delta F^M\over \p v_r}
\biggl({v-\lambda\over \sqrt{D}}\biggr)_{r-k+1}-2 {\p \Delta F^M \over \p u_r} \biggl({1\over \sqrt{D}}\biggr)_{r-k+1}\biggr]
\nn\\
&=D^{-3} e^{u} \bigl(4e^{u}+(v-\lambda)^2\bigr)\nn\\
&\quad - \epsilon^2 \sum_{k,l\geq0} \Biggl[ \frac14 S\bigl(\Delta F^M,v_k,v_l\bigr)
\biggl({v-\lambda\over \sqrt{D}}\biggr)_{k+1} \biggl({v-\lambda\over
\sqrt{D}}\biggr)_{l+1}\nn\\
&\quad -S\bigl(\Delta F^M,v_k,u_l\bigr) \biggl({v-\lambda\over
\sqrt{D}}\biggr)_{k+1} \biggl(\frac1{\sqrt{D}} \biggr)_{l+1}+ S(\Delta F^M,u_k,u_l) \,
\biggl({1\over \sqrt{D}}\biggr)_{k+1} \biggl(\frac1{\sqrt{D}}\biggr)_{l+1}\Biggr]
\nn\\
&\quad -\frac{\epsilon^2} 2 \sum_{k\geq0}  \biggl[ {\p \Delta F^M \over \p v_k}
  \p^{k+1} \biggl({e^{u} {4 e^{u}(v-\lambda) u_1 - ((v-\lambda)^2 + 4 e^{u}) v_1 \over D^3} } \biggr) \nn\\
&\qquad\qquad\quad + {\p \Delta F^M \over \p u_k} \p^{k+1}\biggl(e^{u} {4 (v-\lambda) \, v_1 - ((v-\lambda)^2 + 4 e^{u}) u_1\over D^3}\biggr)\biggr],\nn
\end{align}
where $\Delta F^M:= \sum_{g\geq 1} \epsilon^{2g} F^M_g$,
$D= (v-\lambda)^2 - 4 e^{u}$,
$S(f,a,b):=
{\p^2 f\over \p a \p b}+
{\p f \over \p a}
{\p f \over \p b}$,
and $f_r$ stands for $\p^r(f)$ with 
$$\p:=\sum_{k\geq0} u_{k+1} \frac{\p f}{\p u_k} + \sum_{k\geq0} v_{k+1} \frac{\p f}{\p v_k}.$$
Indeed, according to~\cite{DZ-norm, DZ}, solution $\Delta F^M$ to~\eqref{loop} exists and is unique up to
a sequence of additive constants for $F_g$ ($g\geq1$), and these constants can be further fixed by~\eqref{homogenfg} for $g\ge2$.
For example, the explicit expression for~$F^M_2$ can be found in~\cite{DY1, DZ}. The functions $F^M_g$, $g\geq 2$, also 
satisfy the following equation
\beq
 2 \frac{\p F^M_g}{\p u} + \sum_{k\geq1} v_k \frac{\p F^M_g}{\p v_k} = 0,\label{homogenfg22022}
\eeq
which can be easily deduced from the $L_0$-constraint for the GW invariants of~$\mathbb{P}^1$~\cite{DZ}.

It is shown in~\cite{DZ} that the quasi-Miura map~\eqref{qm1}--\eqref{qm2} 
transforms a monotone solution of the principal hierarchy~\eqref{phdef} to a 
solution of the extended Toda hierarchy (see the Theorem~1.1 of~\cite{DZ}).
As we mentioned above, the particular solution $(v(x,{\bf T};a), u(x, {\bf T};a))$ 
of interest to the $\p_{T^{2,p}}$-flows ($p\geq0$) in the principal hierarchy~\eqref{phdef} specified 
by the initial data~\eqref{initialdata} is monotone, and therefore we conclude that, 
$(\widehat V(x,{\bf T};\e), \widehat U(x,{\bf T};\e))$ obtained by replacing $v_k, u_k$ ($k\geq0$) in~\eqref{qm1}--\eqref{qm2} 
by $\partial_x^k(v(x,{\bf T};a)), \partial_x^k(u(x, {\bf T};a))$, i.e.
\begin{align}
&\widehat V(x,{\bf T};\e) := \widehat V|_{v_k\mapsto\partial_x^k(v(x,{\bf T};a)), u_k\mapsto\partial_x^k(u(x,{\bf T};a)), k\geq0}, \\
&\widehat U(x,{\bf T};\e) := 
\widehat U|_{v_k\mapsto\partial_x^k(v(x,{\bf T};a)), u_k\mapsto\partial_x^k(u(x,{\bf T};a)), k\geq0}
\end{align}
is a particular solution to the Toda lattice hierarchy.  
Since we have shown that $e^{\e^{-2}\F_0(x,{\bf p};a)}=e^{\e^{-2}\widetilde{\F}_0(x,{\bf T};a)}$ is the tau-function of 
the solution $(v(x,{\bf T};a), u(x, {\bf T};a))$ 
to the dispersionless Toda lattice hierarchy, using  
the Theorem~1.1 of~\cite{DZ} we find that
\beq
\tau(x,{\bf T};a; \e):=\exp\Biggl({\e^{-2}\widetilde{\F}_0(x,{\bf T};a)} + \sum_{g\geq1} \e^{2g-2}F^M_g\big|_{v_k\mapsto\partial_x^k(v(x,{\bf T};a)), u_k\mapsto\partial_x^k(u(x,{\bf T};a)), k\geq0} \Biggr)
\eeq
is the tau-function of the solution $\bigl(\widehat V(x,{\bf T};\e), \widehat U(x,{\bf T};\e)\bigr)$ to the Toda lattice hierarchy. 
By using \eqref{dilatonF0}, \eqref{genuszeroEF0}, \eqref{homogMEuv}, \eqref{homogenfg}, \eqref{homogenfg22022}, 
one can also verify that this tau-function $\tau=\tau(x,{\bf T};a; \e)$ satisfies the following two linear equations:
\begin{align}
&\sum_{q\geq0}  \widetilde{T}^{2,q}  \frac{\p \tau}{\p T^{2,q}} + \e  \frac{\p \tau}{\p \e} + x  \frac{\p \tau}{\p x} +
a \frac{\p \tau}{\p a}  + \frac1{12}\tau = 0, \\
&\sum_{q\geq0} (q+1) \widetilde{T}^{2,q}  \frac{\p \tau}{\p T^{2,q}} + \frac{x(x+a)}{\e^2} \tau = 0,
\end{align}
which agree with~\eqref{dilatonforZ121} and the $k=0$ case of~\eqref{virazfull}, respectively. 
Hence we arrive at the following theorem.
\begin{theorem}\label{thm3}
The genus zero corrected dessin/LUE free energy $\F_0(x,{\bf p};a)$ is equal to $\widetilde{\F}_0(x,{\bf T};a)$ given by~\eqref{dessins0dub}.
The genus~$g$ ($g\geq1$) corrected dessin/LUE free energy satisfies
\beq\label{fgfmgequal}
\F_g(x,{\bf p};a) = F^M_g\biggl({\bf v}(x,{\bf T};a), \frac{\p {\bf v}(x,{\bf T};a)}{\p x},\dots,\frac{\p^{3g-2} {\bf v}(x,{\bf T};a)}{\p x^{3g-2}}\biggr) 
+ \frac{c_g}{a^{2g-2}} + \frac{b_g x}{a^{2g-1}}
\eeq
for some $c_g, b_g\in \CC$. Here, $T^{2,q} = q! \, p_{q+1}$,
and
${\bf v}(x,{\bf T};a)=(v(x,{\bf T};a),u(x,{\bf T};a))$ denotes the unique power-series-in-${\bf T}$ solution to
the principal hierarchy~\eqref{phdef} specified by~\eqref{initialdata}.
\end{theorem}
\noindent We verified that for $g=1,2$, 
\beq\label{15141}
c_g=\biggl(\zeta'(-1)-\frac1{24} \log(-1)\biggr)\delta_{g,1}, \quad b_g=0.
\eeq
We expect that the equalities in~\eqref{15141} hold for all $g\geq1$.

Similar statements to Theorems~\ref{thm2}, \ref{thm3} can be given for the case of the
GUE partition function~\cite{Dubrovin, DY1}, i.e., the case with initial value $V(x,{\bf 0};\e) = 0, W(x, {\bf 0};\e) = x$
for the Toda lattice hierarchy. 
The following is a table of the initial values\footnote{We would like to 
thank the authors of~\cite{GGR} for providing the initial value for the JUE case; cf.~\cite{Wilson}.}
 for various tau-functions of the Toda lattice hierarchy (cf.~\cite{DY1, DYZ2, GGR, GGR2, Y, Zhou2}) appearing in matrix models:
\begin{center}
\begin{tabular}{|c|c|c|}
\hline
& $V$ & $W$ \\
\hline
LUE/dessins & $2x+a+\e$ & $x(x+a)$ \\
\hline
JUE & $\frac12+\frac{b^2-a^2}{2(2x+a+b)(2x+a+b+2\e)}$ & $\frac{x(x+a)(x+b)(x+a+b)}{(2x+a+b-\e)(2x+a+b)^2(2x+a+b+\e)}$ \\
\hline
GUE & $0$ & $x$ \\
\hline
mEven GUE & $2x+\frac{\e}2 = x_{\rm mE}+\e$ & $x(x-\frac{\e}2)= \frac14 x_{\rm mE}^2 - \frac{\e^2}{16}$ \\
\hline
$\bP^1$ & $x+\frac{\e}2$  & $1$ \\
\hline
\end{tabular}
\end{center}
where mEven GUE stands for the modified even GUE~\cite{DLYZ, DY2, DY3},
and JUE stands for the Jacobi unitary ensemble.
In each case, let us define the {\it normalized partition function} to be the one that equals~1 when 
all the couplings are taken to be~0.
For the GUE case, we recall from~\cite{Dubrovin, DY1} that a correction factor 
\beq
\e^{-\frac1{12} + \frac{n^2}2} (2\pi)^{-\frac{n}2} G(n+1) 
\eeq
is introduced to the normalized GUE partition function so that the corrected GUE partition function is a tau-function for the Toda lattice hierarchy.
For LUE, the correction factor to the normalized LUE partition function can be read off from~\eqref{defluecorrect}, that is
\beq
\e^{-\frac1{12} + n(n+\alpha)} (2 \pi)^{-n}
\frac{G(n+1) G(n+\alpha+1)}{G(\alpha+1)}.
\eeq 
For JUE, to make the normalized JUE partition function~\cite{GGR2} a tau-function for the 
Toda lattice hierarchy, we need to add the following correction factor (using~\eqref{deftau3} and the second line of the above table)
\begin{align}\label{juecorrection}
\e^{-\frac1{12} + 2n(n+\alpha+\beta)} (2 \pi)^{-2n}\frac{G(n+1) G(n+\alpha+1) G(n+\beta+1) G(n+\alpha+\beta+1)}{G(\alpha+1)G(\beta+1)G(\alpha+\beta+1)}.
\end{align} 
We expect that the higher genus parts of the corrected JUE partition function satisfy Dubrovin--Zhang's ansatz, and we 
have made verifications up to genus~2.

In genus zero we have the following table:
\begin{center}
\begin{tabular}{|c|c|c|c|}
\hline
& $v$ & $w=e^u$ & Equations \\
\hline
LUE/dessins & $2x+a$ & $x(x+a)$ & $w=\frac{1}{4}(v^2-a^2)$ \\
\hline
JUE & $\frac12+\frac{b^2-a^2}{2(2x+a+b)^2}$ & $\frac{x(x+a)(x+b)(x+a+b)}{(2x+a+b)^4}$
& $w=\frac{1}{4}v^2-\frac{b^2}{2(b^2-a^2)}v+\frac{b^2}{4(b^2-a^2)}$  \\
\hline
GUE & $0$ & $x$ & $v\equiv 0$ \\
\hline
mEven GUE &  $2x,\, x_{\rm mE}$  & $x^2,\, \frac14 x_{\rm mE}^2$  & $w=\frac14 v^2$ \\
\hline
$\bP^1$ & $x$  & $1$ & $w\equiv 1$ \\
\hline
\end{tabular}
\end{center}
The last column gives the equations of the algebraic curves
in the $(v,w)$-plane determined by the initial values in genus zero.

We remark that the LUE is just the generalized Penner model called by string theorists (see Appendix~\ref{appendixB}).

To conclude this section,
we recover another result in \cite{GGR} by a different method.
In \cite[Theorem 5.2]{Zhou2} the modified GUE partition function introduced
in \cite{DLYZ} has been identified with the dessin partition function by:
\beq
\widetilde{Z}_{{\rm mE}}\bigl(x_{\rm mE},\{s_{2k}\};\epsilon\bigr) = 
Z_{\rm dessins}\biggl(u=\frac{x_{\rm mE}}2+\frac{\e}{4}, v=\frac{x_{\rm mE}}2-\frac{\e}{4}, \{p_j\}; \e\biggr),
\eeq
where $s_{2k}= \frac{p_k}{2^k k}$, $k\geq1$, and $\widetilde{Z}_{{\rm mE}}(x_{\rm mE}, \{s_{2k}\}; \epsilon)$ denotes 
the {\it normalized modified GUE partition function}, normalized from the modified GUE partition function 
by a factor by requiring that $\widetilde{Z}_{{\rm mE}}(x_{\rm mE}, \{0\}; \epsilon)\equiv 1$.
Now by~\eqref{identityzdzo}, we have
\begin{cor} The normalized modified GUE partition and the normalized LUE partition function are related to each other by
\beq
\widetilde{Z}_{{\rm mE}}\bigl(x_{\rm mE},\{s_{2k}\};\epsilon\bigr)
 = Z_{{\rm LUE1}}\biggl( x = \frac{x_{\rm mE}}2+\frac{\e}{4}, \{p_j\}; - \frac{\e}2; \e\biggr).
\eeq
\end{cor}

This is equivalent to \cite[Theorem 1.5]{GGR}.

\section{Grothendieck's dessin counting and monotone Hurwitz numbers}\label{newsection5}
In this section, we first give a second proof of Corollary~\ref{corzdzo}, 
%  is equal to the normalized LUE partition function, 
  and then study several more consequences of this corollary.

\begin{prop}\label{Zoneweone1}
We have the following formula for the normalized LUE partition function:
\beq\label{zonewone1formula}
Z_{\rm LUE1}(x, {\bf p}; a;\e) = e^{W_{\rm Laguerre}} 1,
\eeq
where $W_{\rm Laguerre}$ is the cut-and-join type linear operator defined by
\beq
W_{\rm Laguerre} := 2 \bigl(x+\frac{a}2\bigr) \Lambda_1^{\rm LUE1} + M_1^{\rm LUE1} + \frac{x(x+a)}{\e^2} p_1 ,
\eeq
with
\begin{align}
& \Lambda_1^{\rm LUE1} = \sum_{i\geq2} (i-2) p_i \frac{\p }{\p p_{i-1}} , \\
& M_1^{\rm LUE1} = \sum_{i\geq2} \sum_{j=1}^{i-1}
\biggl((i-1) p_j p_{i-j} \frac{\p}{\p p_{i-1}} + \e^2 j (i-j) p_{i+1}  \frac{\p^2}{\p p_j \p p_{i-j}}\biggr).
\end{align}
\end{prop}
\begin{proof}
As it is shown in~\cite{AvM,HH}, the power series $Z_n^{\rm LUE1}({\bf p};\alpha;1)$ satisfies the following Virasoro constraints:
\beq
\mathcal{L}^{\rm LUE1}_m Z_n^{\rm LUE1}({\bf p};\alpha;1)  =  0, \quad m\geq0,
\eeq
where
\begin{align}
\mathcal{L}^{\rm LUE1}_m = & \sum_{k=1}^{m-1} k (m-k) \frac{\p^2}{\p p_k \p p_{m-k}}
+ \sum_{k\geq1} (k+m) \tilde p_k \frac{\p }{\p p_{k+m}}  \\ 
& + m (2n+\alpha) \delta_{m\ge1} \frac{\p}{\p p_m} + n(n+\alpha)\delta_{m,0}. \nn
\end{align}
Therefore, $Z_{\rm LUE1}(x,{\bf p};a;\e)$ satisfies that
\beq\label{lonezonea}
L^{\rm LUE1}_m Z_{\rm LUE1}(x,{\bf p};a;\e)  =  0, \quad m\geq0,
\eeq
where
\begin{align}
L^{\rm LUE1}_m = & \e^2 \sum_{k=1}^{m-1} k(m-k)\frac{\p^2}{\p p_k \p p_{m-k}}
+ \sum_{k\geq1} (k+m)\tilde p_k \frac{\p }{\p p_{k+m}} \\
&+ 2m\Bigl(x+\frac{a}2\Bigr) \delta_{m\ge1} \frac{\p}{\p p_m} + \frac{x(x+a)}{\e^2}\delta_{m,0}. \nn
\end{align}
The proof is then the same as the one for~\eqref{zdessinsew1} given 
in~\cite{KZ} (cf.~\cite{Zhou1}), namely, by using~\eqref{lonezonea}, \eqref{dilatonlueone} and~\eqref{zone1}
we obtain~\eqref{zonewone1formula}.
\end{proof}

Corollary~\ref{corzdzo} follows immediately from~\eqref{zdessinsew1} and Proposition~\ref{Zoneweone1}. 
Since we have now a different proof of Corollary~\ref{corzdzo}, one can also use 
 Corollary~\ref{corzdzo} and the theory of orthogonal polynomials 
 to give a different proof of Theorem~\ref{thm2}; 
we omit these details (cf.~\cite{AvM, Deift, DY1, mehta}).

One can deduce from Corollary~\ref{corzdzo} several known facts about LUE correlators by using properties of 
the dessin correlators. For instance,
the variable $x+a=(n+\alpha)\e$ can be viewed as a rescaling of the Wishart parameter $w=n+\alpha$; then the symmetry
\beq n\mapsto n+\alpha, \quad \alpha\mapsto-\alpha \eeq
for the normalized LUE partition function is obvious:
the above transformation corresponds to
\beq
u\mapsto v, \quad v\mapsto u,
\eeq
and we know from \eqref{deffdessins}--\eqref{deffzduvpe} that
$$Z_{\rm dessins}(v,u,{\bf p};\e)=Z_{\rm dessins}(u,v,{\bf p};\e).$$

Another consequence that we deduce from Corollary~\ref{corzdzo} is that
the normalized LUE free energy $\F_{\rm LUE1}(x,{\bf p};a;\e)$ has the genus expansion. Namely, 
by taking logarithms on both sides of~\eqref{identityzdzo}
we get %an immediate consequence of Corollary~\ref{corzdzo} 
\beq
\F_{\rm LUE1}(x,{\bf p};a;\e) =: \sum_{g\geq0} \e^{2g-2} \F_{{\rm LUE1},g}(x,{\bf p};a),
\eeq
where $\F_{{\rm LUE1},g}(x,{\bf p};a)$, $g\geq0$, are power series of ${\bf p}$, called 
the genus~$g$ part of the normalized LUE free energy.

By comparing the logarithmic derivatives of both sides of~\eqref{identityzdzo} and using~\eqref{defdessinscorr}, 
we obtain that for $m\geq1$, and $\mu_1,\dots,\mu_m\geq1$, the derivative 
\beq\label{corcorcor}
\frac{\p^m\F_{\rm LUE1}(x,{\bf p};a;\e)}{\p p_{\mu_1} \dots \p p_{\mu_m}}\bigg|_{{\bf p}={\bf 0}}=
%\prod_{j=1}^m \mu_j \, 
\langle \tau_{\mu_1}\cdots\tau_{\mu_m} \rangle (x,x+a,\e)
\eeq
is a polynomial of $x,a$. By further taking $\e=1$ we find that the connected LUE correlator 
\beq\label{corcor}
 \langle \tr \, M^{\mu_1} \cdots \tr \, M^{\mu_m} \rangle_c=
\prod_{j=1}^m \mu_j \, \langle \tau_{\mu_1}\cdots\tau_{\mu_m} \rangle (n,n+\alpha,1)
\eeq
is a polynomial of $n$ and $\alpha$.
The large $n$ expansion considered by Cunden--Dahlqvist--O'Connell~\eqref{CDOCexpansion}
is then easy to perform, which leads to
an explicit relationship between Grothendieck's dessin counting
and the strict monotone Hurwitz numbers with the precise formula
given in the next corollary.

\begin{cor}\label{cornklhurwitz}
For any $g\geq0$ and a partition $\mu=(\mu_1,\dots,\mu_m)$ of length~$m$,
and for $k,l,m\geq1$ satisfying
$$|\mu|-m-l -k = 2g-2,$$
denoting $n_i$ the multiplicity of~$i$ in~$\mu$,
we have
\beq\label{Nhcorres}
N_{k,l}(\mu) = \frac{\prod_{i=1}^\infty {n_i}!}{|\mu|!}  \sum_{\nu\in \mathcal{P}_{|\mu|} \atop \ell(\nu)=l}  h_g(\mu,\nu).
\eeq
\end{cor}
\begin{proof}
Recalling $u=x=n\e, v=x+a=(n+\alpha)\e$ and taking $\e=1$ in~\eqref{defdessinscorr}
we find
\begin{align}
\langle \tau_{\mu_1} \cdots \tau_{\mu_m} \rangle (n,\alpha,1)
& = \sum_{g\geq0, \, k,l\ge1 \atop 2g+k+l=|\mu|-m+2}  N_{k,l}(\mu)  \, n^k  (n+\alpha)^l  .  \nn
\end{align}
Combing with~\eqref{corcorcor} we find
\begin{align}
n^{m-|\mu|-2}  \langle \tr \, M^{\mu_1} \cdots \tr \, M^{\mu_m} \rangle_c
& = \prod_{j=1}^m \mu_j  \sum_{g\geq0, \, k,l\ge1 \atop 2g+k+l=|\mu|-m+2}  N_{k,l}(\mu) \, n^{k+m-|\mu|-2}  (n+\alpha)^l  \nn\\
%& = \sum_{g\geq0, \, k,l\ge1 \atop 2g+k+l=|\mu|-m+2}  N_{k,l}(\mu)  n^{-2g-l}  (n+\alpha)^l   \nn\\
& = \prod_{j=1}^m \mu_j  \sum_{g\geq0, \, k,l\ge1 \atop 2g+k+l=|\mu|-m+2}  N_{k,l}(\mu) \, n^{-2g}  \Bigl(1+\frac{\alpha}n\Bigr)^l. \nn
\end{align}
The corollary is then proved by taking $n\to\infty$ limit and by comparing with~\eqref{CDOCexpansion}.
\end{proof}

Combinatorially,
counting dessins and strictly monotone Hurwitz numbers are counting 
the numbers of solutions of different equations in the symmetric groups.  
It would be interesting if the above relation~\eqref{Nhcorres} 
could also be obtained from the approach in~\cite{ACEH, BHR, GPH, HO}.

By \eqref{identityzdzo} and formula~\eqref{eqn:Z-A} we immediately get
\begin{cor}
We have the following Schur expansion:
\beq
Z_{\rm LUE1}\bigl(x,{\bf p};a;1\bigr)
= \sum_{\mu \in \cP} s_\mu \, \prod_{\Box\in \mu} \frac{(x+c(\Box))(x+a+c(\Box))}{h(\Box)}.
\eeq
In the fermionic picture,
the dessin tau-function is given by a Bogoliubov transformation:
\beq \label{eqn:Z-A-lue}
Z_{\rm LUE1}\bigl(x,{\bf p};a;1\bigr) = \exp\Biggl(\sum_{m,n \geq 0} A_{m,n} \psi_{-m-\frac{1}{2}}\psi^*_{-n-\frac{1}{2}}\Biggr) |0\rangle,
\eeq
where the coefficients $A_{m,n}$ are given by~\eqref{eqn:Amn} with $u,v$ replaced by $x,x+a$. 
\end{cor}

It is interesting to compare the above formulas with the corresponding formulas for the GUE partition function in \cite{Zhou-hermitian1}. 
By~\eqref{identityzdzo} and Theorem A, %in the Introduction proved in \cite{Zhou2},
 we also have the following formula for generating series of the $m$-point LUE correlators:

\begin{cor}
For each $m\geq2$,
the generating series for $m$-point connected LUE correlators
has the following expression:
\beq 
\sum_{\mu_1,\dots,\mu_m\ge1} 
\frac{\langle \tr \, M^{\mu_1} \cdots \tr \, M^{\mu_m}\rangle_c}
{\lambda_1^{\mu_1+1}\cdots \lambda_m^{\mu_m+1}} = (-1)^{m-1} \sum_{\sigma\in S_m/C_m}
\prod_{j=1}^m \widehat{A}(\lambda_{\sigma(j)},\lambda_{\sigma(j+1)}) - \frac{\delta_{m,2}}{(\lambda-\mu)^2},
\eeq
where $\widehat{A}(\lambda,\mu)$ is defined in~\eqref{defhata} with $w$ replaced by $n+\alpha$.
\end{cor}

\section{Concluding remarks}\label{newsection6}

To conclude, let us summarize our results and  make a conjecture. 
Originally,
a tau-function of the KP hierarchy is obtained by counting the dessins in \cite{KZ},
called the dessin partition function.
In \S~\ref{newsection3} and \S~\ref{newsection5},
this tau-function is identified with the LUE partition function,
and hence is related to the strictly monotone Hurwitz numbers based on the results of \cite{CDOC}.
In \S~\ref{newsection3} it is also shown that 
the corrected dessin partition function is a tau-function of the Toda lattice hierarchy.  
One of its consequence 
is that one can apply the matrix-resolvent method developed in~\cite{DY1} to
compute the dessin correlators. 
Another consequence is that the corrected dessin partition function 
is a tau-function $\tau_{ext-Toda}(\{T^{1,p}, T^{2,p}\}_{p\geq 0})$ of the extended Toda hierarchy,
but with $T^{1,p} = 0$ for $p > 0$.
It is easy to obtain a tau-function $\tau_{2-{\rm dessins},t}(\{T^{1,p},T^{2,p}\}_{p\geq 0})$ with a parameter $t$ 
of the 2D Toda hierarchy~\cite{UT} by counting the dessins.
We conjecture that 
\beq
\lim_{t\to 0} \tau_{2-{\rm dessins},t}(\{T^{1,p},T^{2,p}\}_{p\geq 0}) = \tau_{ext-Toda}(\{T^{1,p}, T^{2,p}\}_{p\geq 0}).
\eeq
This is inspired by a similar phenomenon in the GW theory of $\bP^1$.
In \cite{OP} Okounkov and Pandharipande proved that the equivariant GW invariants of $\bP^1$ 
gives a tau function of the 2D Toda hierarchy. 
This tau function  depends on an additional small parameter $t$. 
The non-equivariant limit, which corresponds to $t \to 0$,
gives the partition function of the ordinary GW invariants of $\bP^1$. 
This partition function was shown to be a  tau-function of the extended Toda hierarchy by Dubrovin and Zhang \cite{DZ}.

\begin{appendices}

\section{A Reflection Formula for the Barnes $G$-Function}

Recall Barnes $G$-function can be defined by the following Weierstrass factorization:
\beq \label{def:G}
G(z+1):=(2\pi)^{z/2}e^{-z(z+1)/2- \frac{1}{2}\gamma z^2} \prod_{n=1}^\infty
\biggl\{ \biggl(1+\frac{z}{n}\biggr)^ne^{-z+z^2/(2n)}\biggr\}.
\eeq
After taking the logarithmic derivative,
one gets:
\beq
 \frac{d}{dz} \log G(z+1)
= \frac{1}{2}\log(2\pi)
- \frac{1}{2} - z +z \frac{d}{dz}\log \Gamma(z+1).
\eeq
Change $z$ to $-z$:
\beq
- \frac{d}{dz} \log G(1-z)
= \frac{1}{2}\log(2\pi)
- \frac{1}{2} + z +z \frac{d}{dz}\log \Gamma(1-z).
\eeq
Add them up:
\ben
\frac{d}{dz} \log \frac{G(1+z)}{G(1-z)}
& = & \log(2\pi) -1 + z \frac{d}{dz}\log (\Gamma(1+z)\Gamma(1-z)) \\
& = & \log(2\pi) -1 + z \frac{d}{dz} \log \frac{\pi z}{\sin (\pi z)} \\
& = & \log(2\pi) - \pi z \cot(\pi z).
\een
After integration we get:
\beq \label{eqn:G-Reflection}
G(1-z) = G(1+z) \, \frac{1}{(2\pi)^z} \, \exp \int_0^z \pi t \cot (\pi t) dt.
\eeq
This formula is attributed to Kinkelin \cite{K},
but Kinkelin's $G$-function is different from Barnes $G$-function.
Kinkelin's $G$-function, denoted by $\tilde{G}(x)$, is defined by:
\beq
\tilde{G}(x):=\exp \biggl(\int_0^x \log \Gamma(t) \; dt + \frac{x(x-1)}{2}
- \frac{x}{2} \log(2\pi)\biggr).
\eeq
It has the following properties:
\ben
&& \tilde{G}(0) = \tilde{G}(1) =1, \\
&& \tilde{G}(x+1) = \tilde{G}(x) x^x \quad \text{for $x>0$},
\een
and
\ben
\tilde{G}(n+1) =1^12^2 \cdots n^n \quad \text{for intergers $n\geq 1$}.
\een

Alternatively, by \eqref{def:G},
\beq \label{def:G(1-z)}
G(1-z):=(2\pi)^{-z/2}e^{z(1-z)/2- \frac{1}{2}\gamma z^2} \prod_{n=1}^\infty
\biggl\{ \biggl(1-\frac{z}{n}\biggr)^ne^{z+z^2/(2n)}\biggr\}.
\eeq
So we get
\beq
\frac{G(1-z)}{G(1+z)} = \frac{e^z}{(2\pi)^z}\prod_{n=1}^\infty
\biggl\{ \biggl(1-\frac{z}{n}\biggr)^n\biggl(1+\frac{z}{n}\biggr)^{-n}e^{2z}\biggr\},
\eeq
and
\beq
G(1-z)G(1+z) = e^{(1+\gamma)z^2}\prod_{n=1}^\infty
\biggl\{ \biggl(1-\frac{z}{n}\biggr)^n\biggl(1+\frac{z}{n}\biggr)^{n}e^{z^2/n}\biggr\}.
\eeq

\section{The Penner model and the generalized Penner model}\label{appendixB}
The original Penner model \cite{P} was used to compute the generating series of orbifold Euler characteristics
of ${\mathcal M}_{g,n}$ by considering the formal expansion of the matrix integrals
\beq
Z(t,N)
=\frac{\int_{\cH_N} dM \exp \Bigl(\frac{1}{t} \, \tr \sum_{m=2}^\infty \frac{1}{m}\bigl(i\sqrt{t}M\bigr)^m\Bigr)}
{\int_{\cH_N} dM \exp \bigl(-\frac{1}{2} \, \tr \, M^2\bigr)}.
\eeq
This can be computed by the method of orthogonal polynomials.
First the integral is converted as follows:
\beq
Z(t,N) = \frac{1}{(2\pi)^{N/2}\prod_{j=1}^Nj!}
\int_{\RR^N} \prod_{\substack{1\leq i,j \leq N\\ i\neq j}}
(x_i-x_j)^2 \prod_{i=1}^N d\mu_t(x_i),
\eeq
where for $t>0$,
\beq
d\mu_t(x) = - i\sqrt{t}(et)^{-t^{-1}}(-z)^{-t^{-1}}e^{-z}dz,
\eeq
for $z= \frac{1}{t}(ix\sqrt{t}-1)$.
By \cite[Lemma 3.5]{P},
the relevant orthogonal polynomials are Laguerre polynomials.
The final result is:
As a formal power series, $Z(N,t)$ is the asymptotic series at $0$ of
\beq
\Biggl( \frac{\sqrt{2\pi t}(et)^{-t^{-1}}}{\Gamma(t^{-1})} \Biggr)^N
\prod_{j=1}^{N-1}(1-jt)^{N-j}.
\eeq
Using Stirling's formula,
one gets the following asymptotic expansion:
\beq
\log Z(t,N)
= \sum_{j=1}^\infty t^j \Biggl( - \frac{N^{j+2}}{j(j+1)(j+2)}
+ \sum_{k=1}^{[(j+1)/2]}(2k-1)
\frac{(j-1)!}{(j+2-2k)!}
\frac{B_{2k}}{(2k)!}
N^{j+2-2k}\Biggr).
\eeq
From this one can derive
\beq
\chi(\cM_{g,n}) = (-1)^n \frac{(2g-3+n)!(2g-1)}{(2g)!} B_{2g}.
\eeq
Because
\beq
\chi({\cM}_{g,n+1}) = (2-2g-n) \, \chi({\cM}_{g,n}),
\eeq
so we also have for $g\geq 2$:
\beq
\chi(\cM_{g,0}) = \frac{B_{2g}}{4g(g-1)}= \frac{\zeta(1-2g)}{2-2g}.
\eeq
It was noted by Distler and Vafa \cite{DV} that
\beq
\sum_{g\geq 2} \chi(\cM_{g,0})z^{2-2g} = \sum_{g\geq 2} \frac{\zeta(1-2g)}{2-2g}z^{2-2g}
\eeq
is the asymptotic series of
\beq
\int_0^z \Bigl(z\frac{d}{dz}\log\Gamma(z+1)\Bigr) dz
-\frac{1}{2} z^2\log z+\frac{1}{4}z^2 - \frac{1}{2} z - \zeta'(-1) +\frac{1}{12}\log z.
\eeq
In \cite{WZ},
the right-hand side was rewritten as
\beq \log G(z+1)
- \biggl( \zeta'(-1) + \frac{z}{2}\log (2\pi) + \biggl(\frac{z^2}{2}-\frac{1}{12}\biggr)\log z
- \frac{3z^2}{4} \biggr).
\eeq
Following Ooguri and Vafa~\cite{OV},
we understand
\bea
&& F_{0,0}(z) = \half z^2\log(z)-\frac{3}{4}z^2, \label{def-F00} \\
&& F_{1,0}(z) = -\frac{1}{12}\log(z) + \zeta'(-1) \label{def-F10}
\eea
as the ``orbifold Euler characteristics of $\cM_{0,0}$ and $\cM_{1,0}$" respectively.
Then we have
\beq
\log G(z+1)
\sim \frac{z}{2}\log (2\pi)+ F_{0,0}(z) + F_{1,0}(z) + \sum_{g \geq 2} \chi(\cM_{g,0})z^{2-2g}.
\eeq
We also have a similar expansion for $\log G(z+x+1)$:
\beq
\begin{split}
& \log G(z+x+1)- \frac{z}{2} \log (2\pi) 
- \frac{1}{2} \log (2\pi)x  \\
\sim &  F_{0,0}(z) + F_{1,0}(z)
+ F_{0,0}'(z)x+ \frac{x^2}{2}F_{0,0}''(z)
+  \sum_{2g-2+n>0} \chi(\cM_{g,n})z^{2-2g-n}x^n.
\end{split}
\eeq
 
Let us explain the rationale for the choice of $F_{0,0}$ and $F_{0,1}$.  
Barnes $G$-function also appears in the study of Chern-Simons on $S^3$.
The path integral in this theory has a factor of $\vol(U(N))$.
Ooguri and Vafa \cite{OV} found:
\beq
\vol(U(N)) = \frac{\sqrt{N} \, (2\pi)^{\frac{1}{2}N^2+\frac{1}{2}N-1}}{G(N+1)}.
\eeq
Then one gets an asymptotic  expansion
\beq
\log(\vol(U(N)) =- \sum_{g\geq 2} \frac{\chi(\cM_{g,0})}{N^{2-2g}}
\eeq
up to some terms coming from genus zero and genus one.
The nonperturbative part of the free energy of Chern-Simons part of the Chern-Simons amplitude of $S^3$ is
\beq
F_{\rm nonpert}
= \log \biggl(e^{\frac{\pi}{8}iN^2}\Big(\frac{2\pi}{k+N}\Big)^{N^2/2}
\frac{G(N+1)}{(2\pi)^{N/2}} \biggr).
\eeq
It has the following expansion (cf.~\cite[(2.12)]{OV}):
\beq
F_{\rm nonpert}
= \frac{1}{2}\lambda_s^{-2}t^2
\biggl(\log(2\pi i t) - \frac{3}{2}\biggr)
- \frac{1}{12}\log\bigl(t\lambda_s^{-1}\bigr) + \zeta'(-1)
+ \sum_{g=2}^\infty  \frac{B_{2g}\lambda_s^{2g-2}}{2g(2g-2)t^{2g-2}},
\eeq
where $\lambda_s=\frac{1}{k+N}$ is the string coupling constant,
and $t=N\lambda_s$ is the t 'Hooft coupling constant.

As noted by Distler and Vafa \cite{DV},
the Penner model is equivalent to the matrix model
\beq
e^F=\int_{\cH_N} e^{Nt \, \tr [\log (1-\phi)+\phi]} d\phi= \int \det(1-\phi)^{Nt} e^{Nt \, \tr\phi}.
\eeq
and one is led to consider
\beq
\int \det(1-\phi)^\alpha e^{Nt \, \tr\phi} = e^{N^2t}\int (\det M)^\alpha e^{-Nt \, \tr M}.
\eeq
To make this well defined,
we should only integrate over  positive-definite $M$.
Generalized Penner model is then defined by
\beq
\int_{\mathcal{H}^+_N} dM (\det M)^\alpha e^{- \frac1\e {\rm tr} \, V(M)}
\eeq
for $V(M) =M - \sum_{i=1}^\infty \frac{1}i p_i M^i$.
This is just the partition function of LUE~\eqref{defznlue1} up to normalization.

As recalled above, % in Appendix~\ref{appendixB}, 
for the Penner model,
one has
\bea
&& F_{0,0}(z)= \half z^2\log(z)-\frac{3}{4}z^2,   \\
&& F_{1,0}(z) = -\frac{1}{12}\log(z) + \zeta'(-1) , \\
&& F_{g,0}(z) = \frac{B_{2g}}{4g(g-1)}z^{2-2g}.
\eea
%For the meaning of $F_{0,0}(z)$, $F_{0,1}(z)$, \dots, see~\cite{P} or Appendix~\ref{appendixB}.
By the Example B1 in~\cite{Du0},
the Legendre type transformation $S_2$ transforms the Frobenius manifold associated with 
Toda lattice with the potential~\eqref{todafrob}
%\beq
%F_{Toda} = \frac{1}{2}v^2u + e^u
%\eeq
to the Frobenius manifold associated with nonlinear Schr\"odinger system with the potential
\beq\label{FNLS416}
F_{\rm NLS}(\varphi,\rho) = \frac{1}{2} \varphi^2\rho + \frac{1}{2}\rho^2 \log \rho - \frac{3}{4}\rho^2,
\eeq
where 
\beq
\varphi = v, \rho = e^u.
\eeq
Dubrovin made an observation in~\cite{Dubrovin} that the GUE partition function coincides 
with a part of the partition function (cf.~\cite{CDZ, CvdLPS, DZ-norm}) of the NLS Frobenius manifold~\eqref{FNLS416}, giving the duality 
between GUE and the $\bP^1$-topological sigma model. 
By taking $\varphi=0$ and $\rho = z$,
we get
\beq
F_{0,0}(z) = F_{\rm NLS}(0,z).
\eeq
This indicates a connection between the generalized Penner model and the $\bP^1$-topological sigma model.

\end{appendices}

\medskip
\medskip

\begin{center}
{\small
Di Yang, School of Mathematical Sciences, University of Science and Technology of China\\
Hefei 230026, P.R.~China\\
e-mail: diyang@ustc.edu.cn\\
~\\
Jian Zhou, Department of Mathematical Sciences, Tsinghua University \\
Beijing 100084, P.R.~China\\
e-mail: jianzhou@mail.tsinghua.edu.cn}
\end{center}

\end{document}